\documentclass[12pt]{article}
\usepackage{amsmath,amssymb,epsf,epsfig,amsthm,times,ifthen,epic,psfrag}
\usepackage{printtime}
\usepackage{enumerate}
\usepackage{cite}    
\usepackage{fancyhdr}
\usepackage{setspace} 
\usepackage{multirow}
\usepackage[center]{caption}
\usepackage{lastpage} 

\def\submitteddate{May 6, 2011}

\begin{document}

\newcommand{\creationtime}{\today \ \ @ \theampmtime}

\pagestyle{fancy}
\renewcommand{\headrulewidth}{0cm}
\chead{\footnotesize{Appuswamy-Franceschetti-Karamchandani-Zeger}}
\rhead{\footnotesize{\submitteddate}}
\lhead{}
\cfoot{Page \arabic{page} of \pageref{LastPage}} 

\renewcommand{\qedsymbol}{$\blacksquare$} 


\newtheorem{theorem}              {Theorem}     [section]
\newtheorem{lemma}      [theorem] {Lemma}
\newtheorem{corollary}  [theorem] {Corollary}
\newtheorem{proposition}[theorem] {Proposition}

\theoremstyle{remark}
\newtheorem{remark}     [theorem] {Remark}
\newtheorem{algorithm}  [theorem] {Algorithm}
\newtheorem{conjecture} [theorem] {Conjecture}

\theoremstyle{definition}         
\newtheorem{definition} [theorem] {Definition}
\newtheorem{example}    [theorem] {Example}
\newtheorem*{claim}  {Claim}
\newtheorem*{notation}  {Notation}
\newcommand{\Comment}[1]{& [\mbox{from  #1}]}
\newcommand{\mc}{\mathcal}
\newcommand{\mb}{\mathbf}
\newcommand{\abs}[1]{\left\lvert #1 \right\rvert}
\newcommand{\card}[1]{\abs{#1}}
\newcommand{\Network}{\mathcal{N}}
\newcommand{\node}{v}
\newcommand{\edge}{e}
\newcommand{\nodes}{\mathcal{V}}
\newcommand{\vertices}[1]{\nodes}
\newcommand{\inNodes}[1]{V_{#1}}
\newcommand{\Capacity}[1]{w(#1)}
\newcommand{\edges}{\mathcal{E}}
\newcommand{\inEdges}[1]{\mathcal{E}_{in}(#1)}
\newcommand{\outEdges}[1]{\mathcal{E}_{out}(#1)}
\newcommand{\encodingFunction}[2]{h_{#1,#2}}
\newcommand{\Dimension}[1]{\text{Dim}(#1)}
\newcommand{\TransVar}[1]{\hat{#1}}
\newcommand{\RecVar}[1]{\tilde{#1}}
\newcommand{\Code}{\mathcal{C}}
\newcommand{\Directed}[1]{\vec{#1}}
\newcommand{\alphabet}{\mathcal{A}}
\newcommand{\sources}{S}
\newcommand{\decodeAlphabet}{\mathcal{B}}
\newcommand{\delay}{\delta}
\newcommand{\integer}{\mathbb{Z}}
\newcommand{\remove}[1]{}
\newcommand{\removeTrue}[1]{}
\newcommand{\sourceSymbol}{\alpha}
\newcommand{\sourceVec}[1]{\sourceSymbol\!\left(#1\right)}
\newcommand{\sumVec}{p}
\newcommand{\edgeVar}[1]{z_{#1}}
\newcommand{\edgeSet}{C}
\newcommand{\cut}[1]{{min-cut{$\left(#1\right)$}}}
\newcommand{\cuts}[1]{\Lambda({#1})}
\newcommand{\minCut}[1]{{min-cut{$\left(#1\right)$}}}
\newcommand{\maxRange}[1]{\hat{R}_{#1}}
\newcommand{\maxRangeB}[1]{R_{#1}}
\newcommand{\receiver}{\rho}
\newcommand{\source}{\sigma}
\newcommand{\sourceVecPrime}[1]{\sourceSymbol'\!\left(#1\right)}
\newcommand{\decodFunct}{\psi}
\newcommand{\setOfMessages}{x}
\newcommand{\vecInst}{\mathbf{\alpha}}
\newcommand{\messageVecInst}{w}
\newcommand{\encodeMatrix}[1]{M_{#1}}
\newcommand{\VecComp}[2]{#1_{#2}}
\newcommand{\Tree}{\mathcal{T}}
\newcommand{\tree}{T}
\newcommand{\treeIndex}[1]{J_{#1}}
\newcommand{\gap}[1]{\mathcal{C}_{\mbox{{\scriptsize gap}}}\!\left(#1\right)}
\newcommand{\codCap}[1]{\mathcal{C}_{\mbox{{\scriptsize cod}}}\!\left(#1\right)}
\newcommand{\linCodCap}[1]{\mathcal{C}_{\mbox{{\scriptsize lin}}}\!\left(#1\right)}
\newcommand{\scalLinCodCap}[1]{\mathcal{C}_{\mbox{{\scriptsize s-lin}}}\!\left(#1\right)}
\newcommand{\routCap}[1]{\mathcal{C}_{\mbox{{\scriptsize rout}}}\!\left(#1\right)}
\newcommand{\edgeFunct}[1]{h^{(#1)}}

\newcommand{\NbinaryBlocks}[1]{{#1}^{(M)}}
\newcommand{\sumset}[1]{\mbox{{\it sum}}\!\left(#1\right)}
\newcommand{\zeroAt}[1]{h^{(#1)}}
\newcommand{\invZeroAt}[1]{h_{#1}^{-1}}
\newcommand{\maxWeight}[1]{\mbox{HW}_{\mbox{\tiny max}}\left(#1\right)}
\newcommand{\cardConst}[1]{\gamma(#1)}
\newcommand{\hammingWeight}[1]{\mbox{HW$\left(#1\right)$}}
\newcommand{\cardsources}{s}
\newcommand{\NOPROCESS}[1]{}
\newcommand{\FunIdentity}{f_{id}}
\newcommand{\FunMajority}{f_{maj}}
\newcommand{\FunSum}{f_{sum}}
\newcommand{\FunParity}{f_{parity}}
\newcommand{\FunMaximum}{f_{max}}
\newcommand{\FunMinimum}{f_{min}}
\newcommand{\footprintsize}{footprint size }
\newcommand{\footprintsizes}{footprint sizes }
\newcommand{\IndicesToIndex}{h}
\newcommand{\steinerNumber}[1]{\Pi\!\left(#1\right)}
\newcommand{\lb}[1]{l\!\left(#1\right)}
\newcommand{\sourceVecList}[2]{\sourceSymbol\!\left(#1\right)_{#2}}
\newcommand{\encodSymbol}{\beta}
\newcommand{\encodCoeff}[2]{\encodSymbol_{#1}^{(#2)}}
\newcommand{\edgeCodSymbol}{\gamma}
\newcommand{\edgeCoeff}[2]{\edgeCodSymbol_{#1}^{(#2)}}
\newcommand{\decodSymbol}{\delta}
\newcommand{\decodCoeff}[2]{\decodSymbol_{#2}^{(#1)}}

\newcommand{\encodCoeffVar}[2]{x_{#1}^{(#2)}}
\newcommand{\edgeCoeffVar}[2]{y_{#1}^{(#2)}}
\newcommand{\decodCoeffVar}[2]{w_{#2}^{(#1)}}

\newcommand{\head}[1]{\textit{head($#1$)}}
\newcommand{\tail}[1]{\textit{tail($#1$)}}
\newcommand{\field}[1]{\mathbb{F}_{#1}}
\newcommand{\ring}[2]{\mathbb{F}_{#1}[#2]}
\newcommand{\receivedVec}{\decodFunct}
\newcommand{\ceil}[1]{\left\lceil #1 \right\rceil}
\newcommand{\mVec}{w}
\newcommand{\code}{\mathcal{C}}
\newcommand{\intAdd}{+}
\newcommand{\modqAdd}{\oplus}
\newcommand{\fieldAdd}{\boxplus}
\newcommand{\linearCode}[1]{$#1$-linear}
\newcommand{\ringOnA}{}
\newcommand{\polyRing}[2]{#1\!\left[ #2 \right]}
\newcommand{\zeroVecOver}[1]{0}
\newcommand{\zeroOver}[1]{0_{#1}}
\newcommand{\sourceSet}[1]{\indexSet_{#1}}
\newcommand{\rank}[1]{\mbox{{\it rank}}\!\left(#1\right)}
\newcommand{\ideal}[1]{\left\langle #1\right\rangle}
\newcommand{\grobner}{Grob$\ddot{\mbox{o}}$ner }
\newcommand{\gbasis}[1]{\mathcal{G}\!\left(#1\right)}
\newcommand{\GBASISFIG}{3}
\newcommand{\GBASISFIGEQ}{4}
\newcommand{\variety}{\mathcal{V}}
\newcommand{\idealSet}{\mathcal{I}}
\newcommand{\rad}[1]{#1^{\mbox{{\footnotesize rad}}}}
\newcommand{\alg}{\mbox{{\footnotesize alg}}}
\newcommand{\indexSet}{I}
\newcommand{\submodule}{\mathcal{M}}
\newcommand{\aperiodic}{{not reducible}}
\newcommand{\periodic}{{reducible}}
\newcommand{\change}[1]{{\color{blue} #1}}
\newcommand{\MatEdge}[2]{M(#1)_{#2}}
\newcommand{\defn}[1]{\textit{ #1}}
%
\begin{titlepage}

\setcounter{page}{0}

\title{Linear Codes, Target Function Classes, \\and Network Computing Capacity%
\thanks{This work was supported by the National Science Foundation 
        and the UCSD Center for Wireless Communications.\newline
\indent The authors are with the Department of Electrical and Computer Engineering, 
        University of California, San Diego, La Jolla, CA 92093-0407. \ \ 
 (rathnam@ucsd.edu,\ massimo@ece.ucsd.edu, \ nikhil@ucsd.edu, \ zeger@ucsd.edu)
}}

\author{Rathinakumar Appuswamy, Massimo Franceschetti,\\
Nikhil Karamchandani, and Kenneth Zeger}

\date{
\textit{IEEE Transactions on Information Theory\\
Submitted: \submitteddate\\
}}

\maketitle

\begin{abstract}
We study the use of linear codes for 
network computing in single-receiver networks
with various classes of target functions of the source messages.
Such classes include 
reducible, injective, semi-injective, and linear target functions over finite fields.
Computing capacity bounds and achievability are given with respect to these
target function classes for
network codes that use routing, linear coding, or nonlinear coding.

%
\end{abstract}

\thispagestyle{empty}
\end{titlepage}

\clearpage

\section{Introduction}
\textit{Network coding} concerns networks where each receiver demands a subset of 
messages generated by the source nodes and the objective is to satisfy the receiver demands
at the maximum possible throughput rate.
Accordingly, research efforts have studied coding gains over routing
~\cite{Ahlswede-Cai-Li-Yeung-IT-Jul00,Harvey1,Harvey2}, 
whether linear codes are sufficient to achieve the capacity 
~\cite{linearnetwork,ken1,ken2,Koetter-Medard-IT-Oct03}, 
and cut-set upper bounds on the capacity 
and the tightness of such bounds 
\cite{Harvey1,Harvey2, AprilLehman-EricLehman-04}.

\textit{Network computing}, 
on the other hand,
considers a more general problem in which 
each receiver node  demands a target function of the source messages
~\cite{Kumar1,computing1,Kumar2,RaiDey2009,ramam,Nazer}.
Most problems in network coding 
are applicable to network computing as well.
Network computing problems arise in various networks including sensor networks and vehicular networks.

In~\cite{computing1}, a network computing model was proposed where the network
is modeled by a directed, acyclic graph with independent, noiseless links.
The sources generate independent messages  and a single receiver node computes a 
target function $f$ of these messages.
The objective is to characterize the
maximum rate of computation, that is, the maximum number of times
$f$ can be computed per network usage.
Each node in the network sends out symbols on 
its out-edges which are arbitrary, but fixed, functions of the symbols received
on its in-edges and any messages generated at the node.
In linear network computing, this encoding is restricted to be linear operations.
Existing techniques for computing in networks 
use routing, where the codeword sent out by a node consists 
of symbols either received by that node, or generated  by the node if it is a source
(e.g.~\cite{Govindan}).

In network coding, 
it is  known that linear codes are sufficient to achieve 
the coding capacity for multicast networks~\cite{Ahlswede-Cai-Li-Yeung-IT-Jul00},
but  they are not sufficient in general to achieve the coding capacity for 
non-multicast networks~\cite{ken1}.
In network computing, 
it is known that 
when multiple receiver nodes demand a scalar linear target function of the source
messages, linear network codes may not be sufficient in general for solvability~\cite{RaiDey2009B}.
However, it has been shown that for single-receiver networks,
linear coding is sufficient for solvability when computing a scalar linear target function
~\cite{isit,RaiDey2009}.
Analogous to the coding capacity for network coding,
the notion of computing capacity was defined for network computing in~\cite{Kumar1} 
and is the supremum of achievable rates of computing the network's target function.

One fundamental objective in the present paper 
is to understand the performance of linear network codes for 
computing different types of target functions.
Specifically, we compare the linear computing capacity with that of the 
(nonlinear) computing capacity and the routing computing capacity
for various different classes of target functions in single-receiver networks.
Such classes include 
reducible, injective, semi-injective, and linear target functions over finite fields.
Informally,
a target function is semi-injective if it uniquely maps at least one of its inputs,
and
a target function is reducible if it can be computed using a linear transformation 
followed by a function whose domain 
has a reduced dimension.
Computing capacity bounds and achievability are given with respect to the
target function classes studied for
network codes that use routing, linear coding, or nonlinear coding.

Our specific contributions will be summarized next.

\subsection{Contributions} \label{Sec:contributions}
Section~\ref{Sec:model}
gives many of the formal definitions used in the paper 
(e.g. target function classes and computing capacity types).
We show
that routing messages through the intermediate nodes 
in a network forces the receiver to obtain all the messages even though only a function of the messages 
is required 
(Theorem~\ref{Th:routingCapacity}),
and we bound the computing capacity gain of using nonlinear versus routing codes
(Theorem~\ref{Th:boundOnCodingGain}).

In Section~\ref{Sec:linear},
we demonstrate
that the performance of optimal linear codes may depend on how
`linearity' is defined
(Theorem~\ref{line}).
Specifically, we show that the linear computing capacity of a network varies depending
on which ring linearity is defined over on the source alphabet.

In Sections~\ref{Sec:LinearCodes} and \ref{Sec:LinearTargetFunctions},
we study the computing capacity gain of using 
linear coding over routing, 
and nonlinear coding over linear coding.
In particular,
we study various classes of target functions,
including
injective, semi-injective, reducible, and linear.
The relationships between these classes is illustrated in 
Figure~\ref{Fig:TargetFunctions}.

Section~\ref{Sec:LinearCodes} studies linear coding for network computing.
We show
that if a target function is not reducible, 
then the linear computing capacity and routing computing capacity are equal whenever 
the source alphabet is a finite field
(Theorem~\ref{Th:linearCodingCapacity});
the same result also holds for semi-injective target functions over rings.
We also show
that whenever a target function is injective,
routing obtains the full computing capacity of a network
(Theorem~\ref{Th:injectTheorem}),
although whenever a target function is neither reducible nor injective,
there exists a network such that the computing capacity is larger than 
the linear computing capacity
(Theorem~\ref{Th:linearCodsWeak}).
Thus for non-injective target functions that are not reducible, any computing capacity 
gain of using coding over routing must be obtained through nonlinear coding.
This result is tight in the sense that if a target function is 
reducible, then  there always exists a network where the linear computing capacity is larger 
than the routing capacity (Theorem~\ref{Th:constantDirAd}).
We also show
that there exists a reducible target function and a network
whose
computing capacity is strictly greater than its linear computing capacity,
which in turn is strictly greater than its routing computing capacity.
(Theorem~\ref{Th:3ineq}).

Section~\ref{Sec:LinearTargetFunctions} focuses on computing linear target functions over finite fields.
We characterize the linear computing capacity for linear target functions over finite fields
in arbitrary networks 
(Theorem~\ref{Th:ModuloSumCodCap}).
We show that linear codes are sufficient for linear target functions
and we upper bound the computing capacity gain of coding (linear or nonlinear)
over routing
(Theorem~\ref{Th:optimalityOfLinearCodes}).
This upper bound is shown to be achievable
for every linear target function and
an associated network,
in which case the computing capacity is equal to the routing computing capacity
times the number of network sources
(Theorem~\ref{Th:linearTargetRoute}).

Finally, Section~\ref{Sec:butterfly}
studies an illustrative example for the computing problem,
namely the reverse butterfly network -- 
obtained by reversing the direction of all the edges in the multicast butterfly network
(the  butterfly network studied in~\cite{Ahlswede-Cai-Li-Yeung-IT-Jul00} illustrated
the capacity gain of network coding over routing).
For this network and the arithmetic sum target function, 
we evaluate the routing and linear computing capacity
(Theorem~\ref{Th:linearCoding})
and the computing capacity
(Theorem~\ref{Th:arithmeticSum}).
We show that the latter is strictly larger than the first two, which are equal to each other. 
No network with such properties is presently known for network coding.
Among other things, the reverse butterfly network also illustrates that 
the computing capacity 
can be a function of 
the coding alphabet 
(i.e. the domain of the target function $f$).
In contrast,
for network coding,
the coding capacity and routing capacity are known to be independent of the coding alphabet used
~\cite{Cannons-Dougherty-Freiling-Zeger05}.

Our main results are summarized in Table~\ref{Tab:Summary}.
\bigskip\bigskip\bigskip\bigskip
\begin{figure}[ht]
\centering
\includegraphics[width=13.5cm]{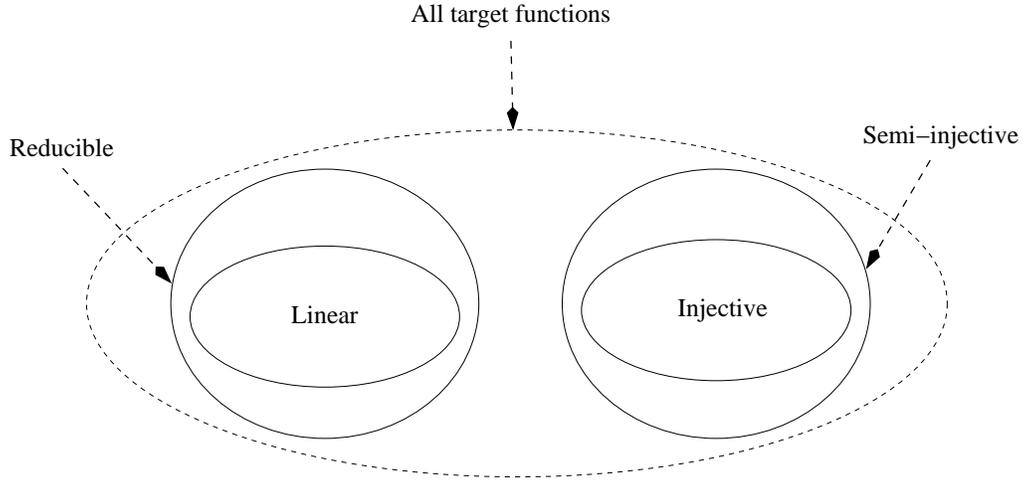}
\caption{
Decomposition of the space of all target functions 
into various classes.
}
\label{Fig:TargetFunctions}
\end{figure}
\bigskip\bigskip\bigskip\bigskip
%
\begin{table}[hht] 
\small
\begin{center}
\renewcommand{\arraystretch}{1.2} 
\begin{tabular}{|c|c|c|c|}
\hline
Result & $f$ & $\alphabet$ & Location\\
\hline 
\hline
\multirow{2}{*}{
\; \ $\forall f\; \forall \Network \  \linCodCap{\Network,f} = \routCap{\Network,f}$} 
& non-reducible & field &
\multirow{2}{*}{Theorem~\ref{Th:linearCodingCapacity}} \\ \cline{2-3}
&
semi-injective & ring & \\
\hline 

\; $\forall f\; \forall \Network \  \codCap{\Network,f} = \routCap{\Network,f}$ &
injective & &
Theorem~\ref{Th:injectTheorem}  \\
\hline 
$\forall f\; \exists \Network \ \codCap{\Network,f} > \linCodCap{\Network,f}$ &
non-injective \& non-reducible & field &
Theorem~\ref{Th:linearCodsWeak} \\
\hline 
\ \ $\forall f\; \exists \Network \ \linCodCap{\Network,f} > \routCap{\Network,f}$ &
reducible & ring &
Theorem~\ref{Th:constantDirAd} \\
\hline 
$\exists f\; \exists \Network \ \codCap{\Network,f} > \linCodCap{\Network,f} > \routCap{\Network,f}$ &
reducible & &
Theorem~\ref{Th:3ineq} \\
\hline 
\ \ $\forall f\; \forall \Network \  \codCap{\Network,f} = \linCodCap{\Network,f} \le s \ \routCap{\Network,f}$ & 
linear & field &
Theorem~\ref{Th:optimalityOfLinearCodes} \\
\hline
$\forall f\; \exists \Network \ \linCodCap{\Network,f} = s \ \routCap{\Network,f}$ &
linear & field &
Theorem~\ref{Th:linearTargetRoute} \\
\hline
$ \exists f\; \exists \Network \ \codCap{\Network,f}$ \ is irrational & 
arithmetic sum  & &
Theorem~\ref{Th:arithmeticSum} \\
\hline
\end{tabular}
\end{center}
\caption{\small Summary of our main results for certain classes of target functions.
The quantities
$\codCap{\Network,f}$,
$\linCodCap{\Network,f}$,
and
$\routCap{\Network,f}$
denote the
computing capacity,
linear computing capacity, 
and
routing computing capacity,
respectively,
for a network $\Network$ with $s$ sources and target function $f$.
The columns labeled $f$ and $\alphabet$ indicate contraints on the target function $f$
and the source alphabet $\alphabet$, respectively.
}
\label{Tab:Summary}
\end{table} 
%

\clearpage
\section{Network model and definitions} \label{Sec:model}
In this paper, a \textit{network} $\Network = (G,\sources,\receiver)$ consists of a finite,
directed acyclic multigraph $G= (\nodes,\edges)$, 
a set 
$\sources = \{\source_1, \dots, \source_\cardsources \} \subseteq \ \nodes$
of $\cardsources$ distinct \textit{source nodes} 
and a single \textit{receiver} $\receiver \in \nodes$. 
We  assume that $\receiver \notin \sources$, 
and that the graph%
\footnote{Throughout the remainder of the paper, 
we use ``graph'' to mean a multigraph, 
and in the context of network computing we use ``network'' to mean a single-receiver network. 
} 
$G$ contains a directed path from every node in $\nodes$ to the receiver $\receiver$. 
For each node $u \in \nodes$, 
let $\inEdges{u}$ and $\outEdges{u}$ denote the in-edges and out-edges of $u$ respectively. 
We assume
(without loss of generality)
that if a network node has no in-edges,
then it is a source node.
If $e=(u,v) \in \edges$, we will use the notation $\head{e}=u$ and $\tail{e}=v$.

An \textit{alphabet}  is a finite set of size at least two. 
Throughout this paper, $\alphabet$ will denote a {\it source alphabet}
and $\decodeAlphabet$ will denote a {\it receiver alphabet}.
For any positive integer $m$,
any vector $x \in \alphabet^{m}$, 
and any  $i \in \{1,2,\ldots,m\}$,
let $\VecComp{x}{i}$ denote the $i$-th component of $x$.
For any index set $\indexSet = \{i_1,i_2,\ldots,i_q\} \subseteq \{1,2,\ldots,m\}$ with $i_1 < i_2 < \ldots < i_q$, let $\VecComp{x}{I}$ denote the vector 
$(\VecComp{x}{i_1},\VecComp{x}{i_2},\ldots,\VecComp{x}{i_q}) \in \alphabet^{\card{I}}$. 
Sometimes we  view $\alphabet$ as an algebraic structure 
such as a ring, i.e., with multiplication and addition.
Throughout this paper, vectors will always be taken to be row vectors.
Let $\field{q}$ denote a finite field of order $q$.
A superscript $t$ will denote the transpose for vectors and matrices.

\subsection{Target functions}
For a given network $\Network = (G, S, \receiver)$, we use $s$ throughout the paper to denote 
the number $\card{S}$ of receivers in $\Network$.
For  given network $\Network$, 
a {\it target function} is a mapping
$$
f : \alphabet^{\cardsources} \longrightarrow \decodeAlphabet.
$$
The goal in network computing is to compute $f$ at the receiver $\receiver$,
as a function of the source messages.
We will assume that all target functions depend on all the network sources
(i.e. a target function cannot be a constant function of any one of its arguments).
Some example target functions that will be referenced are listed in Table~\ref{Tab:exampleFunctions}.
\begin{table}[hht] 
\begin{center}
\renewcommand{\arraystretch}{1.2} 
\begin{tabular}{|c||c|c|c|}
\hline 
Target function  $f$			& 			Alphabet $\mathcal{A}$					&			 $f\left(x_{1}, \ldots , x_{s} \right)$	& 		Comments 	 \\
\hline
\hline 
\defn{identity}				& 		arbitrary   				& $\left(x_{1},  \ldots , x_{s} \right)$	 & $\decodeAlphabet = \alphabet^s$ \\
\hline 
\defn{arithmetic sum}	& 	$\{0,1,\ldots,q-1\}$	& $x_1 + x_2 + \cdots + x_s$	& `$+$' is ordinary integer addition,   \\
\vspace{-0.66cm} \\
& & & $\decodeAlphabet = \{0,1,\cdots,s(q-1)\}$	  \\
\hline
\defn{mod $r$ sum}     &   $\{0,1,\ldots,q-1\}$	& $x_{1} \oplus x_{2} \oplus \ldots \oplus x_{s}$	& $\oplus$ is $\bmod$ $r$ addition, \  $\decodeAlphabet = \alphabet$ \\
\hline
\defn{linear}					&  any ring 								& $a_{1} x_{1} + a_{2} x_2 + \ldots + a_{s} x_{s}$	& 	arithmetic in the ring, \ $\decodeAlphabet = \alphabet$ \\
\hline 
\defn{maximum}					&  any ordered set 				& $\max \left\{x_{1}, \ldots , x_{s} \right\}$	& $\decodeAlphabet = \alphabet$ \\
\hline
\end{tabular}
\end{center}
\caption{Definitions of some target functions.}
\label{Tab:exampleFunctions}
\end{table} 

\begin{definition} \label{Def:linearReducible}
Let alphabet $\alphabet$ be a ring.
A target function $f: \alphabet^s \longrightarrow \decodeAlphabet$ 
is said to be {\it reducible} 
if there exists an integer $\lambda$ satisfying $\lambda < s$,
an  $s \times \lambda$ matrix $T$ with elements in $\alphabet$, and
a map $g: \alphabet^{\lambda} \longrightarrow \decodeAlphabet$ 
such that 
for all $x \in \alphabet^s$,
\begin{align} \label{Eq:linearReducibility}
g(xT)  =  f(x).
\end{align} 
\end{definition}
Reducible target functions are not injective, 
since, for example, if $x$ and $y$ are distinct elements of the 
null-space of $T$, then 
$$
f(x) = g(xT)=g(0) = g(yT) = f(y).
$$ 
\begin{example} \label{Ex:non-Lin1}
Suppose the alphabet is $\alphabet=\field{2}$ and the target function is
$$
f : \field{2}^3 \longrightarrow \{0,1\},
$$
where
$$
f(x) = (x_1+x_2) x_3.
$$
Then, by choosing $\lambda = 2$, 
$$
T = \begin{pmatrix}
1 & 0 \\
1 & 0 \\ 
0 & 1
\end{pmatrix},
$$
and $g(y_1,y_2) = y_1 y_2$, we get
\begin{align*}
g(xT) & = g(x_1+x_2, x_3) \\
		  & = (x_1+x_2)x_3 \\
		  & = f(x).
\end{align*}
Thus the target function $f$ is reducible.
\end{example}
\begin{example} \label{Ex:nonPeriodic}
The notion of reducibility requires that 
for a target function
$f : \alphabet^s \longrightarrow \decodeAlphabet$,
the set $\alphabet$ 
must be a ring.
If we impose any ring structure to the domains of the
identity, arithmetic sum, maximum, and minimum target functions,
then these can be shown
(via our Example~\ref{Ex:semiInj} and Lemma~\ref{Lemma:semiInj})
to be non-reducible.
\end{example}
%

%
%

%
\subsection{Network computing and capacity}
Let $k$ and $n$ be positive integers. 
Given a network $\Network$ with source set $\sources$ and 
alphabet $\alphabet$, a \textit{message generator}
is any mapping
$$\sourceSymbol \ : \ \sources \longrightarrow \alphabet^k.$$
For each source $\source_i \in \sources$, 
$\sourceVec{\source_i}$ is called a \textit{message vector}
and its components  
$$\sourceVec{\source_i}_1, \dots, \sourceVec{\source_i}_k$$
are called \textit{messages}%
\footnote{
For simplicity we assume each source has associated with it exactly one message vector,
but all of the results in this paper can readily be extended to the more general case.}.
\begin{definition}
A $(k,n)$ \textit{network code in a network $\Network$} consists of the following: 
\begin{itemize}
\item[(i)] {\it Encoding functions} $h^{(e)}$,
for every out-edge $\edge \in \outEdges{\node}$
of every node $\node \in \nodes - \receiver$,
of the form: 
\begin{align*}
h^{(e)}: 
&\displaystyle \left(\prod_{\hat{\edge} \in \inEdges{\node}} \alphabet^{n} \right) \times 
 \alphabet^{k} \longrightarrow \alphabet^{n}  \quad \mbox{if $\node$ is a source node}\\
h^{(e)}: 
&\displaystyle \prod_{\hat{\edge} \in \inEdges{\node}} \alphabet^{n} \longrightarrow \alphabet^{n}  
 \hspace{.9in} \mbox{otherwise.}
\end{align*}
\item[(ii)] 
A {\it decoding function} $\decodFunct$ of the form:
$$
\decodFunct: 
\prod_{\hat{\edge} \in \inEdges{\node}} \alphabet^{n} \longrightarrow \decodeAlphabet^{k}.
$$
\end{itemize}
\end{definition}
Furthermore, given a $(k,n)$ network code,
every edge $\edge \in \edges$ carries a vector $z_{\edge}$ of at most $n$ alphabet symbols%
\footnote{By default, we   assume that edges carry exactly $n$ symbols.}, 
which is obtained by evaluating the encoding function 
$h^{(e)}$ on the set of vectors carried by the in-edges to the node and the node's message vector if 
the node is a source. 
The objective of the
receiver is to compute the target function $f$ of the source messages, 
for any arbitrary message generator $\sourceSymbol$.
More precisely, the receiver constructs a vector of $k$ alphabet symbols, 
such that for each $i \in \{1, 2,\ldots, k\}$,
the $i$-th component of the receiver's computed vector equals the value of the desired target function $f$, 
applied to the $i$-th components of the source message vectors,
for any choice of message generator $\sourceSymbol$.
\begin {definition}
Suppose in a network $\Network$,
the in-edges of the receiver are
$e_1, e_2, \ldots, e_{\card{\inEdges{\receiver}}}$.
A $(k,n)$ network code is said to  \textit{ compute $f$ in $\Network$} 
if for each $j \in \{1, 2,\ldots, k\}$, 
and for each message generator $\sourceSymbol$,
the decoding function satisfies 
\begin{align}
\decodFunct\left(\edgeVar{e_1},\cdots,\edgeVar{e_{\card{\inEdges{\receiver}}}}\right)_j 
&= f\!\left((\sourceVec{\source_1}_j,\cdots,\sourceVec{\source_{\cardsources}}_j)\right). \label{Eq:decodingFunction}
\end{align}
If there exists a $(k, n)$  code that computes $f$ in $\Network$, 
then the rational number $k/n$ is said to be an {\it achievable computing rate}.
\end{definition}
%

%
In the network coding literature, 
one definition of the \textit{coding capacity} of a network is the supremum of all achievable coding rates 
~\cite{Cannons-Dougherty-Freiling-Zeger05}.
We use an analogous definition for the computing capacity. 
\begin{definition}
The \textit{computing capacity} of a network $\Network$ with respect to a target function $f$ is
$$
\codCap{\Network,f} \; = \; \sup  \Big\{ \frac{k}{n} \ : \ 
    \mbox{$\exists$  $(k,n)$ network code that computes $f$ in $\Network$}\Big\}. \label{def:cap}
$$
\end{definition}
%
The notion of linear codes in networks is most often studied with respect to finite fields.
Here we will sometimes use more general ring structures.
\begin{definition}
\label{def:linear-code}
Let alphabet $\alphabet$ be a ring.
A $(k,n)$ network code in a network $\Network$ is said to be a \textit{linear network code (over $\alphabet$)}
if the encoding functions
are linear over $\alphabet$.
\end{definition}
\begin{definition}
The \textit{linear computing capacity} of a network $\Network$ with respect to target function $f$ is
$$
\linCodCap{\Network,f} \; = \; \sup  \Big\{ \frac{k}{n} \ : \ 
    \mbox{$\exists$  $(k,n)$ linear network code that computes $f$ in $\Network$}\Big\}. \label{def:linCap}
$$
\end{definition}
The \textit{routing computing capacity $\routCap{\Network,f}$}
is defined similarly by
restricting the encoding functions to routing.
We call the quantity $\codCap{\Network,f}-\linCodCap{\Network,f}$ the {\it computing capacity gain} of using
nonlinear coding  over linear coding.
Similar ``gains'', such as, $\codCap{\Network,f}-\routCap{\Network,f}$ 
and $\linCodCap{\Network,f}-\routCap{\Network,f}$ are defined.

Note that
Definition~\ref{def:linear-code}
allows linear codes to have nonlinear decoding functions.
In fact, 
since the receiver alphabet $\decodeAlphabet$ need not have any algebraic
structure to it, linear decoding functions would not make sense in general.
We do, however, examine a special case where
$\decodeAlphabet = \alphabet$
and the target function is linear,
in which case we show that linear codes with linear decoders can be just
as good as linear codes with nonlinear decoders
(Theorem~\ref{Th:optimalityOfLinearCodes}).

\begin{definition}
A set of edges $C \subseteq \edges$ in network $\Network$
 is said to \textit{separate} 
sources $\source_{m_1}, \ldots, \source_{m_d}$
from the receiver $\receiver$, 
if for each $i \in \{1, 2,\ldots, d\}$,
every directed path from
$\source_{m_i}$ to $\receiver$ contains at least one edge in $C$.
%
Define
\begin{align} 
I_C &= \left\{i : \mbox{$C$ separates $\source_i$ from the receiver} \right\}. \notag
\end{align}
The set $C$ is said to be a \textit{cut} 
in $\Network$ if it separates at least one source from the receiver 
(i.e. $ \card{I_C} \ge 1$).
We denote by $\cuts{\Network}$ the collection of all cuts in $\Network$.
\end{definition}

Since $I_C$ is the number of sources disconnected by $C$ and there are $s$ sources,
we have
\begin{align}\label{eq:IC_ub}
|I_C| \le s.
\end{align}

For network coding with a single receiver node and multiple sources (where the receiver demands
all the source messages), routing is known to be optimal~\cite{AprilLehman-EricLehman-04}.
Let $\routCap{\Network}$ denote the routing capacity of the network
$\Network$, or equivalently the routing computing capacity for computing
the identity target function.
It was observed in~\cite[Theorem~4.2]{AprilLehman-EricLehman-04} that for any single-receiver network $\Network$, 
\begin{align} \label{Eq:routingCap}
\routCap{\Network} = \underset{C \in \cuts{\Network}}{\min} 
 \ \frac{\card{C}}{\card{I_C}}.
\end{align}
The following theorem shows that if the intermediate nodes in a network 
are restricted to perform routing, 
then in order to compute a target function
the receiver is forced to obtain all the 
source messages.
This fact motivates the use of 
coding for computing functions in networks.
\begin{theorem} \label{Th:routingCapacity}
If $\Network$ is a network with target function $f$, then 
%
$$
\routCap{\Network,f} = \routCap{\Network}.
$$
\end{theorem}
\begin{proof}
Since any routing code that computes the identity target function can be used to compute any target function $f$,
we have
$$
\routCap{\Network,f} \ge \routCap{\Network}.
$$
Conversely,
it is easy to see that every component of every source message must be received by $\receiver$
in order to compute $f$, so
$$
\routCap{\Network,f} \le \routCap{\Network}.
$$
\end{proof}

Theorem~\ref{Th:boundOnCodingGain} below gives a general upper bound on how much larger
the computing capacity can be relative to the routing computing capacity.
It will be shown later, in Theorem~\ref{Th:optimalityOfLinearCodes},
that for linear target functions over finite fields,
the bound in Theorem~\ref{Th:boundOnCodingGain}
can be tightened by removing the logarithm term.

\begin{lemma}\label{lem:1}
If $\Network$ is network 
with a target function $f : \alphabet^s \longrightarrow \decodeAlphabet$,
then
\begin{align*}
\codCap{\Network, f} & \le (\log_2 \card{\alphabet})\; \underset{ C \in \cuts{\Network} } \min\ \card{C}.
\end{align*}

\begin{proof}
Using~\cite[Theorem~II.1]{computing1}, 
one finds the term min-cut$(\Network,f)$ defined in~\cite[Equation~(3)]{computing1} in terms of a quantity $R_{I_C,f}$,
which in turn is defined in~\cite[Definition~1.5]{computing1}.
Since target functions are restricted to not being constant functions of any of their arguments,
we have $R_{I_C,f} \ge 2$,
from which the result follows.
\end{proof}

\end{lemma}

\begin{theorem}\label{Th:boundOnCodingGain}
If $\Network$ is network 
with a target function $f : \alphabet^s \longrightarrow \decodeAlphabet$,
then
$$
\codCap{\Network, f} \le  s\;(\log_2\card{\alphabet})\; \routCap{\Network,f}
$$
\end{theorem}

\begin{proof}
\begin{align*}
\codCap{\Network, f} & \le (\log_2 \card{\alphabet}) \; \underset{ C \in \cuts{\Network} } \min\ \card{C} 
                     & \Comment{Lemma~\ref{lem:1}} \\ 
 & \le s\; (\log_2 \card{\alphabet})\; \routCap{\Network,f}. 
       & \Comment{\eqref{eq:IC_ub}, \eqref{Eq:routingCap}, and Theorem~\ref{Th:routingCapacity} }
\end{align*}
\end{proof}

\clearpage

\section{Linear coding over different ring alphabets}
\label{Sec:linear}

Whereas the size of a finite field characterizes the field,
there are, in general, different rings of 
the same size, 
so one must address whether
the linear computing capacity of a network
might depend on which
ring is chosen for the alphabet.
In this section, we illustrate this possibility  with 
a specific computing problem. 

Let $\alphabet = \{a_0,a_1,a_2,a_3\}$
and let  $f: \alphabet^2 \longrightarrow \{0,1,2\}$ be as defined in Table~\ref{Tab:checkFunction}.
\begin{table}[hht] 
\begin{center}
\renewcommand{\arraystretch}{1.2} 
\begin{tabular}{|c||c|c|c|c|}
\hline 
$f$	& $a_0$	& $a_1$	& $a_2$	& $a_3$	 \\
\hline
\hline 
$a_0$	& $0$ & $1$	& $1$	& $2$	 \\
\hline 
$a_1$	& $1$	& $0$	& $2$	& $1$  \\
\hline 
$a_2$	& $1$	& $2$	& $0$	& $1$  \\
\hline 
$a_3$	& $2$	& $1$	& $1$	& $0$ \\
\hline
\end{tabular}
\end{center}
\caption{Definition of the $4$-ary map $f$.}
\label{Tab:checkFunction}
\end{table} 
We consider different rings $R$ of size $4$ for $\alphabet$ and evaluate 
the linear computing capacity of the network
$\Network_1$ shown in Figure~\ref{Fig:2pointNW}
with respect to the target function $f$.
Specifically, we let $R$ be either the ring $\integer_4$ of integers modulo $4$
or the product ring $\integer_2 \times \integer_2$ of $2$-dimensional binary vectors.
Denote the linear 
computing capacity here by
$$
\linCodCap{\Network_1}^{R} = \; \sup  \Big\{ \frac{k}{n} \ : \ 
    \mbox{$\exists$  $(k,n)$ $R$-linear code that computes $f$ in $\Network$}\Big\}.
$$
%
%
%
\begin{figure}[ht]
\begin{center}
\psfrag{X}{$\source_1$}
\psfrag{Y}{$\source_2$}
\psfrag{Z}{$\receiver$}
\scalebox{.9}{\includegraphics{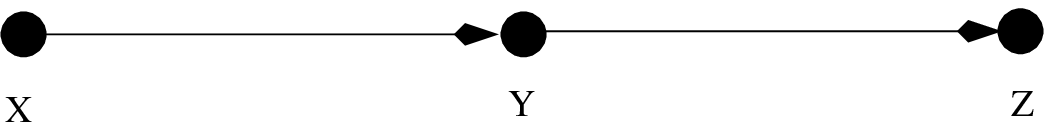}}
\end{center}
\caption{Network $\Network_1$ has two sources $\source_1$ and $\source_2$ and a receiver $\receiver$.} 
\label{Fig:2pointNW}
\end{figure}
The received vector $z$ at $\receiver$ can be viewed as a function of the source vectors generated at 
$\source_1$ and $\source_2$. 
For any $(k,n)$ $R$-linear code, 
there exist $k \times n$ matrices $M_1$ and $M_2$ 
such that $z$ 
can be written as
\begin{align} \label{Eq:linearCodeMatrix_ln}
\edgeVar{}(\sourceVec{\source_1},\sourceVec{\source_2}) =  \sourceVec{\source_1} M_1+ \sourceVec{\source_2} M_2.
\end{align}
Let $m_{i,1},\cdots,m_{i,k}$ 
denote the row vectors of $M_i$,
for $i\in\{1,2\}$.

\begin{lemma} \label{lemma:zerosOfLinearCode_lnw}
Let $\alphabet$ be the ring $\integer_4$
and let $f: \alphabet^2 \longrightarrow \{0,1,2\}$ be the target function shown in 
Table \ref{Tab:checkFunction},
where $a_i = i$, for each $i$.
If a $(k,n)$ linear code over $\alphabet$ \ computes $f$ in $\Network_1$
and $\receiver$ receives a zero vector,
then
$\sourceVec{\source_1} = \sourceVec{\source_2} \in \{0,2\}^k$.
\end{lemma}
\begin{proof}
If $\sourceVec{\source_1}=\sourceVec{\source_2}= 0$,
then $\receiver$ receives a $0$ by \eqref{Eq:linearCodeMatrix_ln}
and must decode a $0$ since
$f((0,0)) = 0$ (from Table~\ref{Tab:checkFunction}).
Thus, $\receiver$ always decodes a $0$ upon receiving a $0$.
But 
$f((x_1,x_2)) = 0 $ if and only if $x_1 = x_2$
(from Table~\ref{Tab:checkFunction}),
so whenever $\receiver$ receives a $0$,
the source messages satisfy
$\sourceVec{\source_1} = \sourceVec{\source_2}$.

Now suppose,
contrary to the lemma's assertion, 
that there exist messages 
$\sourceVec{\source_1}$ and $\sourceVec{\source_2}$
such that
$\edgeVar{}(\sourceVec{\source_1},\sourceVec{\source_2}) = 0$
and
$\sourceVec{\source_1}_j \not\in \{0,2\}$ 
for some $j \in \{1,2,\cdots,k\}$.
Since $\sourceVec{\source_1}_j$
is invertible in $\integer_4$ (it is either $1$ or $3$),
we have from \eqref{Eq:linearCodeMatrix_ln} that
\begin{align}
m_{1,j} & = 
\sum_{\stackrel{i=1}{i \neq j}}^{k} -\sourceVec{\source_1}_j^{-1}\sourceVec{\source_1}_i m_{1,i}  + 
\sum_{i=1}^{k} -\sourceVec{\source_1}_j^{-1} \sourceVec{\source_2}_i m_{2,i}\\
& = y^{(1)} M_1 + y^{(2)} M_2
\label{Eq:123}
\end{align}
where $y^{(1)}$ and $y^{(2)}$ are $k$-dimensional vectors 
defined by
\begin{align}
y^{(1)}_i & = \begin{cases}
		-\sourceVec{\source_1}_j^{-1}\sourceVec{\source_1}_i & \mbox{if $i \neq j$} \\
		0 & \mbox{if $i = j$}
		\end{cases} \notag\\
y^{(2)}_i & = -\sourceVec{\source_1}_j^{-1}\sourceVec{\source_2}_i.	
\label{Eq:defnOfY2}	
\end{align}
Also, define the $k$-dimensional vector $x$ by
\begin{align}
x_i & = \begin{cases}
		0 & \mbox{if $i \neq j$} \\
		1 & \mbox{if $i = j$}.
		\end{cases} \label{Eq:defnOfX} 
\end{align}
We have from \eqref{Eq:linearCodeMatrix_ln} that
$z(x,0) = m_{1,j}$ 
and from \eqref{Eq:linearCodeMatrix_ln} and \eqref{Eq:123} that
$z(y^{(1)},y^{(2)}) = m_{1,j}$.
Thus, in order for the code to compute $f$,
we must have
$f(x_j,0) = f(y^{(1)}_j,y^{(2)}_j)$.
But 
$f(x_j,0) = f(1,0) = 1$ 
and
\begin{align*}
f(y^{(1)}_j,y^{(2)}_j) 
&= f(0,-\sourceVec{\source_1}_j^{-1} \sourceVec{\source_2}_j) \\
&= f(0,-\sourceVec{\source_1}_j^{-1} \sourceVec{\source_1}_j) 
   &\Comment{$\sourceVec{\source_1} = \sourceVec{\source_2}$}\\
&= f(0,-1) \\
&= f(0,3) & \Comment{$3 = -1$ in $\integer_4$} \\
&= 2 & \Comment{Table~\ref{Tab:checkFunction}},
\end{align*}
a contradiction.
Thus, 
$\sourceVec{\source_1} \in \{0,2\}^k$.
\end{proof}
\begin{theorem} \label{line}
The network $\Network_1$ in Figure \ref{Fig:2pointNW} 
with alphabet $\alphabet = \{a_0,a_1,a_2,a_3\}$
and target function $f: \alphabet^2 \longrightarrow \{0,1,2\}$ shown in 
Table~\ref{Tab:checkFunction}, satisfies
\begin{align*}
\linCodCap{\Network_1,f}^{\integer_{4}} &\le \frac{2}{3}\\
\linCodCap{\Network_1,f}^{\integer_2 \times \integer_2 }  &= 1.
\end{align*}
(For $\alphabet=\integer_4$,
we identify $a_i = i$, for each $i$,
and 
for $\alphabet=\integer_2 \times \integer_2$,
we identify each $a_i$ with the $2$-bit binary representation of $i$.)
\end{theorem}
\begin{proof}
Consider a $(k,n)$ \  $\integer_2 \times \integer_2$-linear code that computes $f$.
From \eqref{Eq:linearCodeMatrix_ln},
we have
$\edgeVar{}(x,\zeroVecOver{R}) = \zeroVecOver{R}$
whenever
$x M_1  = \zeroVecOver{R}$.
Since $f((\zeroVecOver{R},\zeroVecOver{R})) \neq f((x_i,\zeroVecOver{R}))$ (whenever $x_i \neq 0$),
it must therefore be the case that
$x M_1  = \zeroVecOver{R}$
only when $x = 0$,
or in other words,
the rows of $M_1$ must be independent,
so $n \ge k$.
Thus,
\begin{equation} \label{Eq:eq1}
\linCodCap{\Network,f}^{\integer_2 \times \integer_2 } \le 1.
\end{equation}
Now suppose that $\alphabet$ is the ring 
$\integer_2 \times \integer_2$ where,
$a_0=(0,0)$, $a_1=(0,1)$, $a_2=(1,0)$, and $a_3=(1,1)$
and let $\oplus$ denote the addition over $\alphabet$. 
For any $x \in \alphabet^2$, 
the value $f(x)$, 
as defined in Table~\ref{Tab:checkFunction},
is seen to be the Hamming distance
between $x_1$ and $x_2$. 
If $k=n=1$
and  $M_1 = M_2 = [a_3]$ (i.e., the $1 \times 1$ identity matrix), 
then $\receiver$ receives $x_1 \oplus x_2$ from which 
$f$ can be computed by summing its components.
Thus, a computing rate of $k/n=1$ is achievable.
From \eqref{Eq:eq1}, it then follows that
$$\linCodCap{\Network,f}^{\integer_2 \times \integer_2 } = 1.$$

We now  prove that 
$\linCodCap{\Network,f}^{\integer_{4}} \le 2/3$.
Let $\alphabet$ denote the ring $\integer_{4}$
where $a_i = i$ for $0\le i \le 3$.
For a given $(k,n)$ linear code over $\alphabet$ that computes $f$, 
the $n$-dimensional vector received by $\receiver$ can be
written as in \eqref{Eq:linearCodeMatrix_ln}.
Let $\mathcal{K}$ denote the collection of all message vector
pairs $(\sourceVec{\source_1}, \sourceVec{\source_2})$ such that
$z(\sourceVec{\source_1},\sourceVec{\source_2}) = 0$.
Define the $2k \times n$ matrix
$$M = \begin{bmatrix} M_1 \\ M_2 \end{bmatrix}$$
and notice that $\mathcal{K} = \{ y \in \alphabet^{2k} : y M = 0\}.$
Then,
\begin{align*}
4^n   
 & = \card{\alphabet}^n \\ 
 & \ge \card{ \{ y M : y \in \alphabet^{2k} \} } 
  & \Comment{$y\in \alphabet^{2k} \Longrightarrow yM\in \alphabet^n$} \\
 &  \ge \frac{\card{\alphabet}^{2k}}{\card{\mathcal{K}}}  
  & \Comment{$y^{(1)},y^{(2)} \in \alphabet^{2k}$ and  $y^{(1)}M =y^{(2)}M$ $\Longrightarrow$ $y^{(1)}-y^{(2)} \in \mathcal{K}$}\\
 & \ge \frac{\card{\alphabet}^{2k}}{2^k} 
  & \Comment{Lemma~\ref{lemma:zerosOfLinearCode_lnw}} \\
 & = 4^{3k/2}. & \Comment{$\card{\alphabet} = 4$}
\end{align*}
Thus, $k/n \le 2/3$,
so 
$\linCodCap{\Network_1,f}^{\integer_{4}} \le \frac{2}{3}$.
\end{proof}

%
%

%
\clearpage
%
\section{Linear network codes for computing target functions} \label{Sec:LinearCodes}
Theorem~\ref{Th:routingCapacity} showed that if intermediate network nodes use routing, 
then a network's receiver learns all the source  messages irrespective 
of the target function it demands. 
In Section~\ref{Sec:nonPeriodic}, 
we prove a similar result when the intermediate nodes use linear network coding.
It is shown that whenever a target function is not reducible the linear computing capacity coincides with the routing capacity
and the receiver must learn all the source messages.
We also show that there exists a network such that the computing capacity is larger than 
the routing capacity whenever the target function is non-injective. 
Hence, if the target function is not reducible, 
such capacity gain must be obtained from nonlinear coding.
Section~\ref{Sec:periodic} shows that linear codes may provide a computing capacity gain over routing 
for reducible target functions and that linear codes may not suffice to obtain 
the full computing capacity gain over routing.
\subsection{Non-reducible target functions} \label{Sec:nonPeriodic}
Verifying whether or not a given target function is reducible may not be easy.
We now define a class of 
target functions that are easily shown to not be reducible.
\begin{definition}
A target function $f : \alphabet^s \longrightarrow \decodeAlphabet$ is said to be \textit{semi-injective} if 
there exists $x \in \alphabet^s$
such that
$f^{-1}(\{f(x)\}) = \{x\}$.
\end{definition}

Note that injective functions are semi-injective.

\begin{example} \label{Ex:semiInj}
If $f$ is the  arithmetic sum target function,
then $f$ is semi-injective 
(since $f(x)=0$ implies $x=0$)
but not injective
(since $f(0,1) = f(1,0) = 1$).
Other examples of semi-injective target functions include the identity, maximum, and minimum functions.
\end{example}


\begin{lemma} \label{Lemma:semiInj}
If alphabet $\alphabet$ is a ring,
then semi-injective target functions are not reducible.
\end{lemma}
\begin{proof}
Suppose that a target function $f$ is reducible.
Then there exists
an integer $\lambda$ satisfying $\lambda < s$,  
matrix $T \in \alphabet^{s \times \lambda}$,
and map $g:\alphabet^\lambda \longrightarrow \decodeAlphabet$ 
such that
\begin{align} \label{Eq:red21}
g(xT) & = f(x) \ \mbox{ for each }  x \in \alphabet^s.
\end{align}
Since $\lambda < s$, there exists a non-zero $d \in \alphabet^s$ such that $d T = 0$. 
Then for each $x \in \alphabet^s$, 
\begin{align}
%
f(d + x) & = g((d + x)T) = g(xT) = f(x)
\label{Eq:AltCondn2}                                     
\end{align}
so $f$ is not semi-injective.
\end{proof}

\begin{definition} \label{Def:perp}
Let $\alphabet$ be a finite field
and let $\submodule$ be a subspace of the vector space $\alphabet^s$ over the scalar field $\alphabet$.
Let
$$
\submodule^{\perp} = \left\{ y \in \alphabet^s : xy^t = 0 \; \mbox{for all $x \in \submodule$} \right\}
$$
and let $\dim(\submodule)$ denote the dimension
of $\submodule$ over $\alphabet$.
\end{definition}
\begin{lemma}%
\footnote{
This lemma is a standard result in coding theory regarding dual codes over finite fields,
even though the operation $xy^t$ is not an inner product
(e.g. \cite[Theorem 7.5]{Hill-book} 
or 
\cite[Corollary 3.2.3]{Nebe-Rains-Sloane}). 
An analogous result for orthogonal complements over 
inner product spaces is well known in linear algebra
(e.g. \cite[Theorem 5 on pg. 286]{Hoffman-Kunze}). 
}
\label{Lemma:perp}
If  $\alphabet$ is a finite field and  $\submodule$ is a subspace of vector space  $\alphabet^s$,
then $(\submodule^{\perp})^{\perp}$ = $\submodule$.
\end{lemma}
%
%
%
%
%
%
%
%

Lemma~\ref{Lemma:equivalence} 
will be used in Theorem~\ref{Th:linearCodingCapacity}.
The lemma states an alternative 
characterization of reducible target functions when the source alphabet is a finite field
and of semi-injective target functions when the source alphabet is a group.

\begin{lemma} \label{Lemma:equivalence}
Let $\Network$ be a network with 
target function 
$f : \alphabet^{\cardsources} \longrightarrow \decodeAlphabet$
and alphabet $\alphabet$.
\begin{itemize}
\item[(i)]
Let $\alphabet$ be a finite field.
$f$ is reducible if and only if 
there exists a non-zero $d \in \alphabet^s$ 
such that for each $a \in \alphabet$ and each $x \in \alphabet^s$,
\begin{align*} 
f(ad + x)  =  f(x).
\end{align*} 

\item[(ii)]
Let $\alphabet$ be a group.
$f$ is semi-injective if and only if
there exists $x \in \alphabet^s$
such that for every non-zero $d \in \alphabet^s$,
\begin{align*} 
f(d + x) \ne  f(x).
\end{align*} 
\end{itemize}
(The arithmetic in $ad+x$ and $d+x$ is performed component-wise over the corresponding $\alphabet$.)
\end{lemma}

\begin{proof}
(i)
If $f$ is reducible, then there exists
an integer $\lambda$ satisfying $\lambda < s$,  
matrix $T \in \alphabet^{s \times \lambda}$,
and map $g:\alphabet^\lambda \longrightarrow \decodeAlphabet$ 
such that
\begin{align} \label{Eq:red1}
g(xT) & = f(x) \ \mbox{ for each }  x \in \alphabet^s.
\end{align}
Since $\lambda < s$, there exists a non-zero $d \in \alphabet^s$ such that $d T = 0$. 
Then for each $a \in \alphabet$ and each $x \in \alphabet^s$, 
\begin{align}
%
f(a d + x) & = g((ad + x)T) = g(xT) = f(x).
\label{Eq:AltCondn}					 
\end{align}
%
Conversely, suppose that there exists a non-zero $d$ such that 
\eqref{Eq:AltCondn} holds for every $a \in \alphabet$ and every $x \in \alphabet^s$
and let $\submodule$ be the one-dimensional subspace of $\alphabet^{\cardsources}$ spanned by $d$.
Then 
\begin{equation}
\label{Eq:FuncSubspace}
f(t + x) = f(x) \ \mbox{ for every } t \in \submodule, x \in \alphabet^s.
\end{equation}
Note that $\dim(\submodule^{\perp}) = \cardsources -1$. 
Let $\lambda = \cardsources - 1$, 
let $T \in \alphabet^{s \times \lambda}$ be a matrix such that its columns 
form a basis for $\submodule^{\perp}$, 
and let $\mathcal{R}_T$ denote the row space of $T$.
Define the map 
$$
g : \mathcal{R}_T \longrightarrow f(\alphabet^s) 
$$
as follows. 
For any $y \in \mathcal{R}_T$ such that $y =  xT$ for $x \in \alphabet^s$, 
let
\begin{align} \label{Eq:mapConstruct}
g(y) = g(xT) = f(x).
\end{align}
%
Note that if $y = x^{(1)}T = x^{(2)}T$ for $x^{(1)} \neq x^{(2)}$,
then 
\begin{align}
  (x^{(1)} - x^{(2)})  T & = 0 \nonumber \\
x^{(1)}-x^{(2)} & \in (\submodule^{\perp})^{\perp} & \Comment{construction of $T$} \nonumber \\
x^{(1)}-x^{(2)} & \in \submodule & \Comment{Lemma~\ref{Lemma:perp}} \nonumber \\
f(x^{(1)})  & = f((x^{(1)}-x^{(2)}) + x^{(2)}) \nonumber \\
 				& = f(x^{(2)}). & \Comment{\eqref{Eq:FuncSubspace}} \nonumber
\end{align}
Thus $g$ is well defined. 
Then from \eqref{Eq:mapConstruct} and Definition~\ref{Def:linearReducible}, 
$f$ is reducible.
\remove{
Also, 
\begin{align*}
g(xT) & = f(x) \quad \forall \ x \in \alphabet^s & \Comment{\eqref{Eq:mapConstruct}} \\
	 T  & \in \alphabet^{s \times (s-1)}. & \Comment{$\dim(\mathcal{R}_T)=s-1$}
\end{align*}
}

\medskip

(ii)
Since $f$ is semi-injective, 
there exists a  $x \in \alphabet^s$ such that $\{x\} = f^{-1}(\{f(x)\})$,
which in turn is true if and only if
for each non-zero $d \in \alphabet^s$, 
we have $f(d + x)  \neq  f(x)$.

\end{proof}
The following example shows 
that if the alphabet $\alphabet$ is not a finite field,
then the assertion in Lemma~\ref{Lemma:equivalence}(i)
may not be true.
\begin{example} \label{Ex:counterExample}
Let $\alphabet = \integer_4$,
let $f : \alphabet \longrightarrow \alphabet$ be the target function defined by $f(x) = 2x$,
and let $d = 2$.
Then,
for all $a \in \alphabet$,
\begin{align*}
f(2a+x)  & = 2 (2a+x) \\
		& = 2x & \Comment{$4=0$ in $\integer_4$}\\
		& = f(x)
\end{align*}
but,
$f$ is not reducible,
since $s = 1$.
\end{example}
%
Theorem~\ref{Th:linearCodingCapacity} establishes for a
network with a finite field alphabet, 
whenever the target function is not reducible,
linear computing capacity is equal to the routing computing capacity,
and therefore
if a linear network code is used, 
the receiver ends up learning all the source messages even 
though it only demands a function of these messages.

For network coding 
(i.e. when $f$ is the identity function),
many multi-receiver networks have a larger linear capacity than their routing capacity.
However,
all single-receiver networks are known to achieve their coding
capacity with routing~\cite{AprilLehman-EricLehman-04}.
For network computing,
the next theorem shows that
with non-reducible target functions
there is no advantage to using linear coding over routing.%
\footnote{
As a reminder, ``network'' here refers to single-receiver networks in the context of computing.
}

\begin{theorem} \label{Th:linearCodingCapacity}
Let $\Network$ be a network with 
target function 
$f : \alphabet^{\cardsources} \longrightarrow \decodeAlphabet$
and alphabet $\alphabet$.
If
$\alphabet$ is a finite field and $f$ is not reducible,
or
$\alphabet$ is a ring with identity and $f$ is semi-injective,
then
$$
\linCodCap{\Network,f} =  \routCap{\Network,f}.
$$
\end{theorem}

\begin{proof}
Since any routing code is in particular a linear code,
%
$$
\linCodCap{\Network,f} \ge \routCap{\Network,f}.
$$
Now consider
a $(k,n)$ linear code that computes the target function $f$ in $\Network$
and let $C$ be a cut.
We will show that for any two collections of source messages,
if the messages agree at sources not separated from $\receiver$ by $C$
and the vectors agree on edges in $C$,
then there exist two other source message collections with different target function values,
such that the receiver $\receiver$ cannot distinguish this difference.
In other words, the receiver cannot properly compute the target function in the network.

For each $e \in C$, there exist $k \times n$ matrices 
$\MatEdge{e}{1}, \ldots, \MatEdge{e}{\cardsources}$ such that the vector carried on $e$ is
$$
\sum_{i=1}^\cardsources \sourceVec{\source_i} \MatEdge{e}{i}.
$$
For any matrix $M$, denote its $j$-th column by $M^{(j)}$.
Let $w$ and $y$ be different $k\times s$ matrices over $\alphabet$,
whose $j$-th columns agree for all $j \notin I_C$.

Let us suppose that the vectors carried on the edges of $C$,
when the the column vectors of $w$ are the source messages,
are the same as
when the the column vectors of $y$ are the source messages.
Then,
for all $e \in C$,
\begin{align} \label{Eq:matrixLnr}
\sum_{i=1}^\cardsources w^{(i)} \MatEdge{e}{i} 
= \sum_{i=1}^\cardsources y^{(i)} \MatEdge{e}{i}.
\end{align}
We will show that this 
leads to a contradiction,
namely that $\receiver$ cannot compute $f$.
Let $m$ be an integer such that if
$d$ denotes the $m$-th row of $w-y$, 
then $d \neq 0$.
For the case where $\alphabet$ is a field and $f$ is not reducible,
by Lemma~\ref{Lemma:equivalence}(i), there exist
$a \in \alphabet$ 
and $x \in \alphabet^s$
such that $ad \ne 0$ and
\begin{align} \label{Eq:contradiction}
f(ad + x) \neq f(x).
\end{align}
In the case where $\alphabet$ is a ring with identity and $f$ is semi-injective,
we obtain \eqref{Eq:contradiction} from Lemma~\ref{Lemma:equivalence}(ii) in the special case of $a=1$.

Let $u$ be any $k\times s$ matrix over $\alphabet$ 
whose $m$-th row is $x$  and let $v = u+a(w-y)$.
From \eqref{Eq:contradiction}, 
the target function $f$ differs on the $m$-th rows of $u$ and $v$.
Thus, 
the vectors on the in-edges of the receiver $\receiver$ 
must differ between two cases:
(1) when the sources messages are the columns of $u$,
and
(2) when the sources messages are the columns of $v$.
The vector carried by any in-edge of the receiver is a function of
each of the message vectors $\sourceVec{\source_j}$,
for  $j \notin I_C$,
and the vectors carried by the edges in the cut $C$. 
Furthermore, the $j$-th columns of $u$ and $v$ agree if $j \notin I_C$.
Thus, at least one of the vectors on an edge in $C$ 
must change when the set of source message vectors 
changes from $u$ to $v$. 
However this is contradicted by the fact that
for all $e \in C$,
the vector carried on $e$ when the columns of $u$ are the source messages is
\begin{align} 
\nonumber
\sum_{i=1}^{\cardsources} u^{(i)} \MatEdge{e}{i}  
& = \sum_{i=1}^{\cardsources} u^{(i)} \MatEdge{e}{i}  
+ a \sum_{i=1}^{\cardsources}   (\messageVecInst^{(i)}-y^{(i)})) 
\MatEdge{e}{i}   &\Comment{\eqref{Eq:matrixLnr}} \\
	& = \sum_{i=1}^{\cardsources} v^{(i)} \MatEdge{e}{i} 
\label{Eq:edgeVecEq}
\end{align}
which is also
the vector carried on $e$ when the columns of $v$ are the source messages.

Hence, for any two different matrices $w$ and $y$ whose $j$-th columns agree for all
$j \notin I_C$, 
at least one vector carried by an edge in the cut $C$ has to differ in value
in the case where the source messages are the columns of $w$
from
the case where the source messages are the columns of $y$.
This fact implies that
$$
(\card{\alphabet}^n)^{\card{C}} \ge (\card{\alphabet}^k)^{\card{I_C}}
$$
and thus
\begin{align*}
\frac{k}{n} & \le \frac{\card{C}}{\card{I_C}}. 
\end{align*}
Since the cut $C$ is arbitrary, 
we conclude (using \eqref{Eq:routingCap}) that 
\begin{align*} 
\frac{k}{n} & \le \underset{C \in \cuts{\Network}}{\min} 
 \ \frac{\card{C}}{\card{I_C}} = \routCap{\Network,f}.
\end{align*}
Taking the supremum over all $(k,n)$ linear network codes that compute $f$ in $\Network$,
we get
\begin{align*} 
\linCodCap{\Network,f} & \le \routCap{\Network,f}.
\end{align*}
\remove{
We claim that the edge vector on at least one of the edges in $C$ changes 
when the message generator changes from 
$\alpha(\source_i) = w^{(i)}$ (for $i \in \{1,2,\cdots,s\}$)
to 
$\alpha(\source_i) = y^{(i)}$ (for $i \in \{1,2,\cdots,s\}$).

We prove this claim by contradiction.
Suppose that the linear code assigns the same vectors on
all edges in $C$ for both cases. Then 
\begin{align}
\sum_{i=1}^{\cardsources} \messageVecInst^{(i)} \encodeMatrix{j,i}  & = \sum_{i=1}^{\cardsources} y^{(i)} \encodeMatrix{j,i} & \mbox{for all $j \in \{1,2,\ldots,\card{C}\}$} \nonumber \\
\sum_{i=1}^{s} (\messageVecInst^{(i)} - y^{(i)}) \encodeMatrix{j,i} & = 0 & \mbox{for all $j \in \{1,2,\ldots,\card{C}\}$}. \label{Eq:zeroCutVec}
\end{align}
Since $w$ and $y$ are distinct,
for some $m \in \{1,2,\ldots,k\}$, 

%
For each $i \in \{1,2,\ldots,s\}$, choose $v^{(i)},u^{(i)} \in \alphabet^k$  
such that 
\begin{align*}
\left( (v^{(1)})_j,(v^{(2)})_j,\ldots,(v^{(\cardsources}))_j \right) & = x \; \, \mbox{for all} \;  j\ \in \ \{1,2,\ldots,k\} \\
 u^{(i)} & = v^{(i)} + a(\messageVecInst^{(i)}-y^{(i)}) \; \, \mbox{for all} \;  i \ \in \ \{1,2,\ldots,s\}
\end{align*}

From \eqref{Eq:zeroCutVec} we also have
%
\begin{align} 
\sum_{i=1}^{\cardsources}  (v^{(i)} + a(\messageVecInst^{(i)}-y^{(i)})) \encodeMatrix{j,i} =  \sum_{j=1}^{\cardsources} v^{(j)} \encodeMatrix{1,j} \quad \mbox{for all $j \in \{1,2,\ldots,\card{C}\}$}. \label{Eq:edgeVecEq}
\end{align}
This implies that each edge in $C$ carries the same vector when  $\sourceVec{\source_i} = v^{(i)}$ (for all $i \in \{1,2,\ldots,s\}$) 
as well as when $\sourceVec{\source_i} = u^{(i)}, i \in \{1,2,\ldots,s\}$. 
Furthermore, $u^{(i)} = v^{(i)}$ for all $ i \notin I_C$ (this follows  from the fact that $y^{(i)} = \messageVecInst^{(i)}$ for every $i \notin I_C$). 
Since the code is linear, every vector carried by an in-edge of the receiver is a matrix-linear combination of
the messages $\left\{ \sourceVec{\source_j} : j \notin I_C \right\}$ and the vectors carried 
by the edges in $C$.
Consequently the vector carried by each in-edge of the receiver does not change when 
the message generator changes from  $\sourceVec{\source_i} = v^{(i)}$ 
to $\sourceVec{\source_i} = u^{(i)}$.
In order for the code to compute $f$, we must have 
$$
f((u^{(1)})_j,(u^{(2)})_j,\ldots,(u^{(s)})_j) = f((v^{(1)})_j,(v^{(2)})_j,\ldots,(v^{(s)})_j) \quad \mbox{for all $j \in \{1,2,\ldots,k\}$}
$$
which contradicts \eqref{Eq:contradiction} because 
$$
((u^{(1)})_m,(u^{(2)})_m,\ldots,(u^{(s)})_m) = ad+x \; \mbox{and} \; ((v^{(1)})_m,(v^{(2)})_m,\ldots,(v^{(s)})_m)=x.
$$
Thus at least one edge in $C$ carries a different vector when one of the messages in 
$\{\sourceVec{\source_i}, i \in I_C\}$ changes and hence
$$
(\card{\alphabet}^n)^{\card{C}} \ge (\card{\alphabet}^k)^{\card{I_C}}
$$
and thus
\begin{align}
\frac{k}{n} & \le \frac{\card{C}}{\card{I_C}}. \label{Eq:uBoundLinCod}
\end{align}
Since the cut $C$ was arbitrary, we conclude that 
\begin{align} \label{Eq:linCodCap}
\frac{k}{n} & \le \underset{C \in \cuts{\Network}}{\min} 
 \ \frac{\card{C}}{\card{I_C}}. & \Comment{\eqref{Eq:uBoundLinCod}}
\end{align}
The statement of the theorem now follows from 
\eqref{Eq:routingCap} and \eqref{Eq:linCodCap}.
}
\end{proof}

\begin{figure}[ht]
\begin{center}
\psfrag{x0}{$\node$}
\psfrag{x1}{$\source_1$}
\psfrag{x2}{$\source_2$}
\psfrag{x3}{$\source_{s-1}$}
\psfrag{x4}{$\source_s$}
\psfrag{T}{$\receiver$}
\psfrag{L}{$s-1$}
\scalebox{.7}{\includegraphics{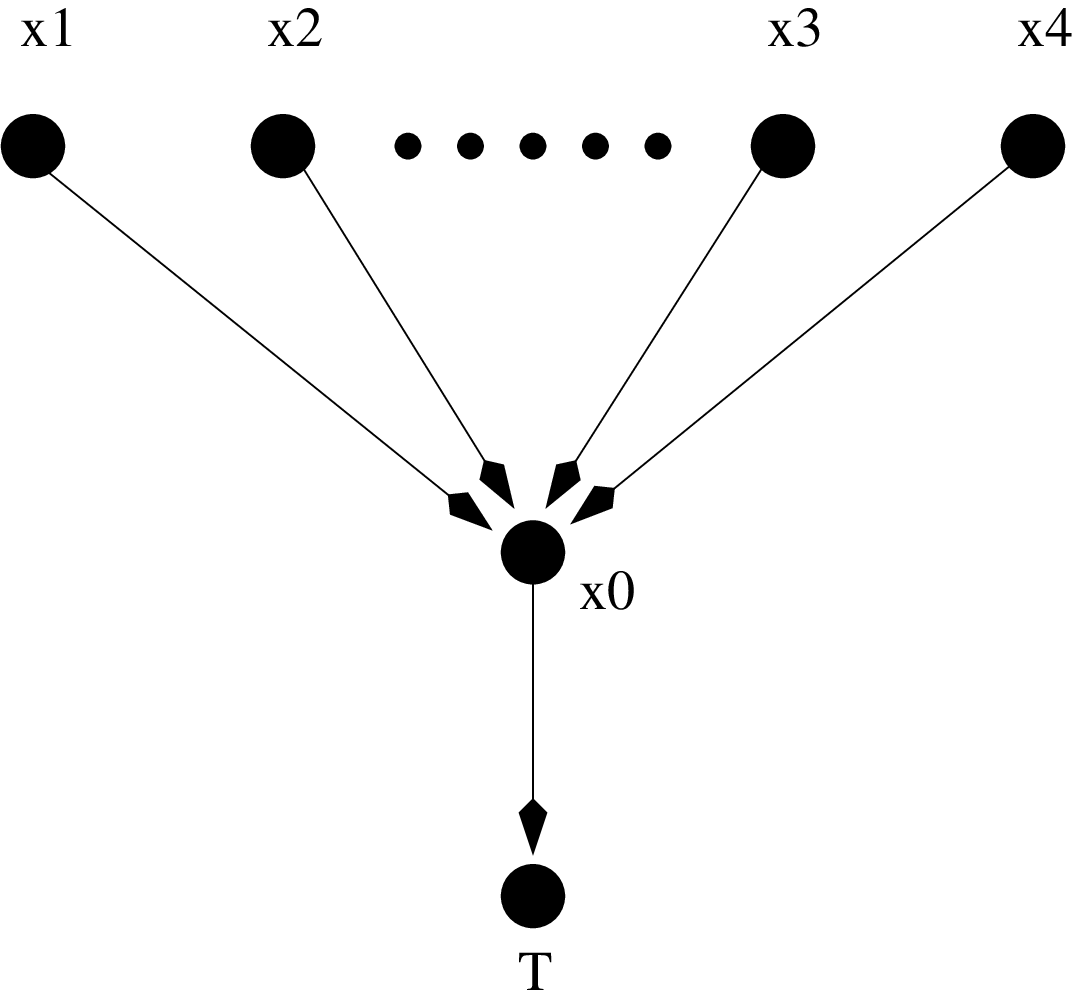}}
\end{center}
\caption{Network $\Network_{2,s}$ has sources
     $\source_{1}, \source_{2},\ldots, \source_{s}$, each connected to the relay 
$\node$ by an edge
and $\node$ is connected to the receiver by an edge.
     } \label{Fig:counterExample}
\end{figure}

Theorem~\ref{Th:linearCodingCapacity}
showed that if a network's target function is not reducible
(e.g. semi-injective target functions)
then there can be no 
computing capacity gain of using linear coding over routing.
The following theorem shows that
if the target function is injective,
then there cannot even be any nonlinear computing gain over routing.

Note that if the identity target function is used in Theorem~\ref{Th:injectTheorem},
then the result states that there is no coding gain over routing for ordinary network coding.
This is consistent since our stated assumption in Section~\ref{Sec:model}
is that only single-receiver networks are considered here
(for some networks with two or more receivers,
it is well known that
linear coding may provide network coding gain over network routing).
\begin{theorem} \label{Th:injectTheorem}
If $\Network$ is a network with an injective target function $f$, then
$$
\codCap{\Network,f} = \routCap{\Network,f}.
$$
\end{theorem}
\begin{proof}
It follows from~\cite[Theorem 4.2]{AprilLehman-EricLehman-04}
that for any single-receiver network $\Network$ and the identity target function $f$, 
we have $\codCap{\Network,f} = \routCap{\Network, f}$.
This can be straightforwardly 
extended to injective target functions for network computing.
\end{proof}

Theorem~\ref{Th:linearCodingCapacity}
showed that there cannot be linear computing gain 
for networks whose target functions are not reducible,
and 
Theorem~\ref{Th:injectTheorem}
showed that the same is true
for target functions that are injective.
However, Theorem~\ref{Th:linearCodsWeak} will show via an example network that 
nonlinear codes may provide a capacity gain 
over linear codes if the target function is not injective. 
This reveals a limitation of 
linear codes compared to nonlinear ones for non-injective target functions that are not reducible. 
For simplicity, in Theorem~\ref{Th:linearCodsWeak} we only consider the case when there are two or more sources.
We need the following lemma first.

\begin{lemma} \label{Lemma:counterExample}
The computing capacity of the network
$\Network_{2,s}$ shown in Figure~\ref{Fig:counterExample},
with respect to a target function $f :  \alphabet^s \longrightarrow \decodeAlphabet$,
satisfies
$$
\codCap{\Network_{2,s},f} \ge  \min \left\{ 1, \; \frac{1}{\log_{\abs{\alphabet}} \abs{f\left( \alphabet^s\right)}} \right\}.
$$
\end{lemma}
\begin{proof}
Suppose
\begin{align}\label{Eq:range1}
\log_{\abs{\alphabet}} \abs{f\left( \alphabet^s\right)} < 1.
\end{align}
Let $k=n=1$ and assume that each source node
sends its message to node $\node$.
Let 
$$
g \ : \ f\left( \alphabet^s\right) \longrightarrow \alphabet 
$$
be any injective map (which exists by \eqref{Eq:range1}).
Then the node $\node$ can compute $g$ and send it to the receiver.
The receiver can compute the value of $f$ from the value of $g$ 
and thus a rate of $1$ is achievable,
so 
$\codCap{\Network_{2,s},f} \ge 1$.

Now suppose
\begin{align} \label{Eq:rangeBound}
\log_{\abs{\alphabet}} \abs{f\left( \alphabet^s\right)} \ge 1.
\end{align}
Choose integers $k$ and $n$ such that 
\begin{align} \label{Eq:rateChoise}
\frac{1}{\log_{\abs{\alphabet}} \abs{f\left( \alphabet^s\right)}} - \epsilon \le \frac{k}{n} \le \frac{1}{\log_{\abs{\alphabet}} \abs{f\left( \alphabet^s\right)}}.
\end{align}
Now choose an arbitrary injective map (which exists by
\eqref{Eq:rateChoise})
$$
g \ : \ (f\left( \alphabet^s\right))^k \longrightarrow \alphabet^n.
$$
%
Since $n \ge k$ (by \eqref{Eq:rangeBound} and \eqref{Eq:rateChoise}), we can still assume that each source sends its $k$-length message vector to node $\node$.
Node $\node$ computes $f$ for each of the $k$ sets of source messages,
 encodes those values into an $n$-length
vector over $\alphabet$ using the injective map $g$ and transmits it to the receiver.
The existence of a decoding function which satisfies \eqref{Eq:decodingFunction} is then obvious 
from the fact that $g$ is injective.
From \eqref{Eq:rateChoise}, the above code  achieves
a computing rate of
$$
\frac{k}{n} \ge \frac{1}{\log_{\abs{\alphabet}} \abs{f\left( \alphabet^s\right)}} - \epsilon.
$$
Since $\epsilon$ was arbitrary, it follows that the computing
capacity $\codCap{\Network_{2,s},f}$ is at least 
$1/\log_{\abs{\alphabet}} \abs{f\left( \alphabet^s\right)}$.
\end{proof}
%

\begin{theorem} \label{Th:linearCodsWeak}
Let $\alphabet$ be a finite field alphabet.
Let $s \ge 2$ and let $f$ be a target function that is neither injective
nor reducible.
Then there exists a network $\Network$ such that
$$
\codCap{\Network,f} > \linCodCap{\Network,f}.
$$
\end{theorem}
\begin{proof}
If $\Network$ is the network $\Network_{2,s}$ shown in Figure~\ref{Fig:counterExample} with alphabet $\alphabet$,
then
\begin{align*}
\linCodCap{\Network,f} & = 1/s & \Comment{Theorem~\ref{Th:linearCodingCapacity} and \eqref{Eq:routingCap}} \\
														 & < \min \left\{ 1, \; \frac{1}{\log_{\abs{\alphabet}} \abs{f\left( \alphabet^s\right)}} \right\} & \Comment{$s \ge 2$ and $\abs{f\left( \alphabet^s\right)} < \abs{\alphabet}^s$} \\
														 & \le \codCap{\Network,f}. & \Comment{Lemma~\ref{Lemma:counterExample}}
\end{align*}
\end{proof}

The same proof of Theorem~\ref{Th:linearCodsWeak} 
shows that it also holds if the alphabet $\alphabet$ is
a ring with identity and the target function $f$ is semi-injective but not injective.

\subsection{Reducible target functions} \label{Sec:periodic}
In Theorem~\ref{Th:constantDirAd},
we prove a converse to Theorem~\ref{Th:linearCodingCapacity} 
by showing that if a target function is reducible, then there exists a network in which the
linear computing capacity is larger than the routing computing capacity.
Theorem~\ref{Th:3ineq} shows that,
even if the target function is reducible,
linear codes may not achieve the full (nonlinear) computing capacity of a network.
\begin{theorem} \label{Th:constantDirAd}
Let $\alphabet$ be a ring.
If a target function $f :  \alphabet^s \longrightarrow \decodeAlphabet$ 
is reducible,
then there exists a network $\Network$
such that 
$$
\linCodCap{\Network,f} > \routCap{\Network,f}.
$$
\end{theorem}
\begin{proof}
Since $f$ is reducible, 
there exist $\lambda < s$, 
a matrix $T \in \alphabet^{s \times \lambda}$,
and a map $g : \alphabet^{\lambda} \longrightarrow f(\alphabet^s)$ such that
\begin{align} \label{Eq:linearConstruct}
g(xT) = f(x) \ \mbox{ for every }  x \in \alphabet^s. & \Comment{Definition~\ref{Def:linearReducible}}
\end{align}
Let $\Network$ denote the network $\Network_{2,s}$ with alphabet $\alphabet$
and target function $f$. 
Let $k = 1$, $n = \lambda$ and let the decoding function
be $\decodFunct = g$.
Since $n \ge 1$, we assume that all the source nodes transmit their messages to node $\node$.
For each source vector 
$$
x = (\sourceVec{\source_1}, \sourceVec{\source_2}, \ldots, \sourceVec{\source_s})
$$
node $\node$ computes  $xT$ and sends it to the receiver.
Having received the $n$-dimensional vector $xT$, the receiver computes
\begin{align*}
\decodFunct(xT) & = g(xT) & \Comment{$\decodFunct = g$} \\
						    & = f(x).  & \Comment{\eqref{Eq:linearConstruct}}
\end{align*}
Thus there exists a linear code that computes $f$ in $\Network$ with an 
achievable computing rate of 
\begin{align*}
\frac{k}{n} & = \frac{1}{\lambda} \\
		        & > 1/s & \Comment{ $\lambda \le s-1$} \\
	         	& = \routCap{\Network} & \Comment{\eqref{Eq:routingCap}}
\end{align*}
which is sufficient to establish the claim.
\end{proof}
For target functions that are not reducible, 
any improvement on achievable rate of computing using coding must be provided by nonlinear codes (by 
Theorem~\ref{Th:linearCodingCapacity}).
However, within the class of reducible target
functions, 
it turns out that there are target functions for which linear 
codes are optimal (i.e., capacity achieving) as shown in Theorem~\ref{Th:optimalityOfLinearCodes}, 
while for certain other
reducible target functions, nonlinear codes might provide a strictly larger
achievable computing rate compared to linear codes.
\begin{remark}
It is possible for a network $\Network$ to have a reducible target function $f$ but 
satisfy $\linCodCap{\Network,f} = \routCap{\Network,f}$ since the network topology  may not
allow coding to exploit the structure of the target function to obtain a capacity gain.
For example, the 3-node network in Figure~\ref{Fig:2pointNW1} with $f(x_1,x_2) = x_1 + x_2$ 
and finite field alphabet $\alphabet$ has
$$
\linCodCap{\Network,f} = \routCap{\Network,f} = 1.
$$
\begin{figure}[ht]
\begin{center}
\psfrag{X}{$\source_1$}
\psfrag{Z}{$\source_2$}
\psfrag{Y}{$\receiver$}
\scalebox{.9}{\includegraphics{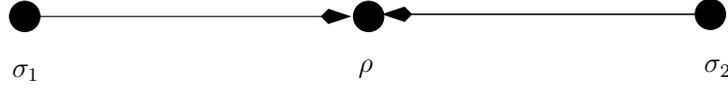}}
\end{center}
\caption{ A network where there is no benefit to using linear coding over routing for computing $f$.} \label{Fig:2pointNW1}
\end{figure}

\end{remark}

Theorem~\ref{Th:linearCodsWeak} shows that for every non-injective, non-reducible target function,
some network has a nonlinear computing gain over linear coding,
and 
Theorem~\ref{Th:constantDirAd} shows that for every reducible
(hence non-injective)
target function,
some network has a linear computing gain over routing.
The following theorem shows that for some reducible target function,
some network has both of these linear and nonlinear computing gains.

\begin{theorem} \label{Th:3ineq}
There exists a network $\Network$ and a reducible target function $f$ such that:
$$
\codCap{\Network,f} > \linCodCap{\Network,f} > \routCap{\Network,f}.
$$
\end{theorem}
\begin{proof}
Let $\Network$ denote the network $\Network_{2,3}$ shown in Figure~\ref{Fig:counterExample} with $s = 3$, alphabet $\alphabet=\field{2}$,
and let $f$ be the target function in Example~\ref{Ex:non-Lin1}.
%
%
%
%
The routing capacity is given by
\begin{align} \label{Eq:routCap1}
\routCap{\Network,f} & = 1/3. & \Comment{\eqref{Eq:routingCap}}
\end{align}
Let $k = n =1$. 
Assume that the sources send their respective  messages to node $\node$.
The target function $f$ can then be computed at $\node$ and sent to the receiver.
Hence, 
$k/n = 1$ is an achievable computing rate and thus
\begin{align} \label{Eq:codingCap1}
\codCap{\Network,f} \ge 1.
\end{align}
Now consider any $(k,n)$ linear code that computes
$f$ in $\Network$.
Such a linear code immediately implies a $(k,n)$ linear code that computes the
target function
$g(x_1,x_2) = x_1 x_2$ in  network $\Network_{2,2}$ as follows.
From the $(k,n)$ linear code that computes
$f$ in $\Network$, we get a $3k \times n$ matrix $M$ such that the node $\node$ in network $\Network$
computes
$$
 \begin{pmatrix} 
			\sourceVec{\source_1}  &
			\sourceVec{\source_2} &
			\sourceVec{\source_3} 
		\end{pmatrix} \ M
$$
and the decoding function computes $f$ from the resulting vector.
Now, in $\Network_{2,2}$, we let the node $\node$ compute
$$
\begin{pmatrix} 
			\sourceVec{\source_1}  &
			0 &
			\sourceVec{\source_2}
		\end{pmatrix} \ M
$$
and send it to the receiver.
The receiver can compute the function $g$ from the received $n$-dimensional vector using the relation $g(x_1,x_2) = f(x_1,0,x_2)$.
Using the fact that the function $g$ is not reducible (in fact, it is semi-injective),  
\begin{align*}
\frac{k}{n} & \le \linCodCap{\Network_{2,2},g} \\
					  & = \routCap{\Network_{2,2},g} & \Comment{Theorem~\ref{Th:linearCodingCapacity}} \\
					  & = 1/2. & \Comment{\eqref{Eq:routingCap}}
\end{align*}					  
Consequently, 
\begin{align} \label{Eq:linCodBnd1}
\linCodCap{\Network,f} \le 1/2.
\end{align}
Now we will construct a $(1,2)$ linear code that computes $f$ in $\Network$.
Let $k = 1$, $n=2$ and 
$$
M = \begin{pmatrix}
			1 & 0 \\
			1 & 0 \\
			0 & 1  
		\end{pmatrix}.
$$
Let the sources send their respective messages to $\node$ while $\node$ computes 
$$
 \begin{pmatrix} 
			\sourceVec{\source_1}  &
			\sourceVec{\source_2} &
			\sourceVec{\source_3} 
		\end{pmatrix} \ M
$$
and transmits the result to the receiver from which $f$ is computable.
Since the above code achieves a computing rate of $1/2$, 
combined with \eqref{Eq:linCodBnd1}, we get
\begin{align} \label{Eq:linCodCap1}
\linCodCap{\Network,f} = 1/2.
\end{align}
The claim of the theorem now follows from 
\eqref{Eq:routCap1}, 
\eqref{Eq:codingCap1}, and 
\eqref{Eq:linCodCap1}.
\end{proof}
\clearpage
\section{Computing linear target functions} \label{Sec:LinearTargetFunctions}

We have previously shown that for reducible target functions there may be a computing capacity gain
for using linear codes over routing.
In this section, we show that for a special subclass of reducible target functions,
namely linear target functions%
\footnote{The definition of ``linear target function'' was given in Table~\ref{Tab:exampleFunctions}.
}
over finite fields,
linear network codes achieve the full (nonlinear) computing capacity.
We now describe a special class of linear codes over finite fields
that suffice for computing linear target functions over finite fields at the maximum possible rate.

Throughout this section, 
let $\Network$ be a network and
let $k$, $n$, and $c$ be positive integers such that $k/n =c$.
Each $k$ symbol message vector generated by a source $\source \in \sources$
can be viewed as a $c$-dimensional vector
%
$$
\sourceVec{\source} = 
(\sourceVec{\source}_1, \sourceVec{\source}_2, \ldots, \sourceVec{\source}_c)
\in \field{q^k}
$$
where $\sourceVec{\source}_i \in \field{q^n}$ for each $i$.
Likewise, the decoder $\decodFunct$ generates a vector of $k$ symbols from $\field{q}$,
which can be viewed as 
a $c$-dimensional vector of symbols from $\field{q^n}$.
For each $\edge \in \edges$,
the edge vector $z_{\edge}$ is viewed as an 
element of $\field{q^n}$.


For every node $u \in \nodes - \receiver$, 
and every out-edge $\edge \in \outEdges{u}$, 
we choose an encoding function $h^{(e)}$
whose output is:
\begin{align}
& \begin{cases}
\displaystyle \sum_{\hat{e} \in \inEdges{u}} \edgeCoeff{\hat{e}}{e} \edgeVar{\hat{e}} 
+ 
\displaystyle \sum_{j=1}^{c} \encodCoeff{j}{e} \sourceVecList{u}{j} 
& \; \text{if} \; u \in \sources \\
\displaystyle \sum_{\hat{e} \in \inEdges{u}} \edgeCoeff{\hat{e}}{e} \edgeVar{\hat{e}} 
& \; \text{otherwise}
\end{cases} \label{Eq:encoders}
\end{align}
for some $\edgeCoeff{\hat{e}}{e}, \encodCoeff{j}{e} \in \field{q^n}$
and we use a decoding function $\decodFunct$ 
whose $j$-th component output $\decodFunct_j$ is:
\begin{align}
& \displaystyle \sum_{e \in \inEdges{\receiver}} \decodCoeff{e}{j} \edgeVar{e} 
\quad \mbox{for all $j \in \{1,2,\ldots,c\}$}
\label{eq:DecodingFunction}
\end{align}
for certain $\decodCoeff{e}{j} \in \field{q^n}$. 
Here we 
view each $h^{(e)}$ as a function of the in-edges to $e$ 
and the source messages generated by $u$ and we
view $\psi$ as a function of the inputs to the receiver.
The chosen encoder and decoder are seen to be linear.

Let us denote the edges in $\edges$ by
$e_1, e_2, \dots, e_{\card{\edges}}$.
For each source $\source$ and each edge $e_j \in \outEdges{\source}$,
let $\encodCoeffVar{1}{e_j}, \dots, \encodCoeffVar{c}{e_j}$ be variables,
and 
for each $e_j \in \inEdges{\receiver}$,
let $\decodCoeffVar{e_j}{1} , \dots, \decodCoeffVar{e_j}{c}$ be variables.
For every $e_i, e_j \in \edges$ such that $\head{e_i}=\tail{e_j}$,
let $\edgeCoeffVar{e_i}{e_j}$ be a variable.
Let $x,y,w$ be vectors containing all the variables
$\encodCoeffVar{i}{e_j}$, 
$\edgeCoeffVar{e_i}{e_j}$,
and
$\decodCoeffVar{e_j}{i}$,
respectively. We will use the short hand notation $\field{}[y]$ to mean the ring of polynomials 
$\field{}[\cdots,y_{e_i}^{(e_j)},\cdots]$
and similarly for $\field{}[x,y,w]$.

Next, we define matrices $A_\tau(x)$, $F(y)$, and $B(w)$.
\begin{itemize}
\item [(i)] For each $\tau \in \{1,2,\cdots,s\}$,
let $A_\tau(x)$ be a $c \times \card{\edges}$ matrix $A_{\tau}(x)$, 
given by
\begin{align} \label{Eq:MatrixA}
\left(A_{\tau}(x)\right)_{i,j} = 
\begin{cases}
\encodCoeffVar{i}{e_j}  & \; \text{if} \; e_j \in \outEdges{\source_{\tau}} \\
0 									& \; \text{otherwise}
\end{cases}
\end{align}
\item [(ii)] Let $F(y)$ be a $\card{\edges} \times \card{\edges}$ matrix,
given by
\begin{align} \label{Eq:MatrixF}
(F(y))_{i,j} = 
\begin{cases}
\edgeCoeffVar{e_i}{e_j}   & \; \text{if} \; e_i, e_j \in \edges \mbox{ and } \head{e_i}=\tail{e_j} \\
0 								  	& \; \text{otherwise}
\end{cases}
\end{align}
\item [(iii)] Let $B(w)$ be a $c \times \card{\edges}$ matrix,
given by
\begin{align} \label{Eq:MatrixB}
(B(w))_{i,j} = 
\begin{cases}
\decodCoeffVar{e_j}{i}  & \; \text{if} \; e_j \in \inEdges{\receiver} \\
0 							    & \; \text{otherwise}.
\end{cases}
\end{align}
\end{itemize}
Consider an $(nc,n)$ linear code of the form in \eqref{Eq:encoders}--\eqref{eq:DecodingFunction}.

Since the graph $G$ associated with the network is acyclic, 
we can assume that the edges $e_1,e_2,\ldots$ 
are ordered such that the matrix 
$F$ is strictly upper-triangular,
and thus we can apply Lemma~\ref{Lemma:inverse}.
Let $I$ denote the identity matrix of suitable dimension.

\begin{lemma}(Koetter-M\'{e}dard~\cite[Lemma~2]{Koetter-Medard-IT-Oct03})
\label{Lemma:inverse}
The matrix $I-F(y)$ is invertible over the ring
$\polyRing{\field{q}}{y}$.
\end{lemma}

\begin{lemma}(Koetter-M\'{e}dard\cite[Theorem~3]{Koetter-Medard-IT-Oct03})
\label{Lemma:receivedSum}
For $s=1$ and for all $\tau\in\{1, \dots, s\}$, 
the decoder in \eqref{eq:DecodingFunction} satisfies
$$
\receivedVec = \sourceVec{\source_{1}} A_{\tau}(\beta) (I - F(\gamma))^{-1} B(\delta)^{t}.
$$
\end{lemma}

%
\begin{lemma}(Alon\cite[Theorem~1.2]{Alon})
\label{Lemma:Alon}
Let $\field{}$ be an arbitrary field, and let $g = g(x_1,\ldots,x_m)$ be a polynomial in 
$\polyRing{\field{}}{x_1,\ldots,x_m}$.
Suppose the degree $deg(g)$ of $g$ is \ $\sum_{i=1}^{m} t_i$, where  each $t_i$ is a nonnegative integer, 
and suppose the coefficient  of \  $\prod_{i=1}^{m} x_i^{t_i}$ in $g$ is nonzero.
Then, if \ $S_1,\ldots,S_m$ are subsets of \ $\field{}$ with $\card{S_i} > t_i$, there are 
$s_1 \in S_1$, $s_2 \in S_2, \ldots, s_m \in S_m$ so that 
$$
g(s_1,\ldots,s_m) \neq 0.
$$
\end{lemma}

For each $\tau \in \left\{1,2,\ldots,s\right\}$,
define the $c\times c$ matrix
\begin{equation} \label{Eq:defnOfMtau}
M_{\tau}(x,y,w)=A_{\tau}(x) (I - F(y))^{-1} B(w)^{t}
\end{equation}
where the components of $M_{\tau}(x,y,w)$
are viewed as lying in $\field{q}[x,y,w]$.

\begin{lemma} \label{Lemma:coeffChoice}
If \ for all  $\tau \in \{1,2,\ldots,s\}$,
$$\det \left( M_{\tau}(x,y,w) \right) \ne 0$$
in the ring $\field{q}[x,y,w]$,
then there exists an integer $n > 0$ and vectors
$\encodSymbol, \edgeCodSymbol, \decodSymbol$ over $\field{q^n}$ such that
for all  $\tau \in \{1,2,\ldots,s\}$ 
the matrix
$M_{\tau}(\encodSymbol,\edgeCodSymbol,\decodSymbol)$
is invertible in the ring of $c\times c$ matrices with components in $\field{q^n}$.
\end{lemma}
\begin{proof}
The quantity 
\begin{align*}
\det \left( \prod_{\tau =1}^{s} M_{\tau}(x,y,w) \right)
\end{align*}
is a nonzero polynomial in   $\field{q}[x,y,w]$ and therefore also in  
$\field{q^n}[x,y,w]$ for any $n \ge 1$. 
Therefore, we can choose $n$ large enough such that the degree of this polynomial is less than $q^n$.
For such an $n$, Lemma~\ref{Lemma:Alon} 
implies there exist
vectors $\encodSymbol, \edgeCodSymbol, \decodSymbol$ 
(whose components correspond to the components of the vector variables $x, y, w$)
over $\field{q^n}$ 
such that
\begin{align} \label{Eq:alon}
\det \left( \prod_{\tau =1}^{s} M_{\tau}(\encodSymbol, \edgeCodSymbol, \decodSymbol) \right) \neq 0.
\end{align}
and therefore, for all  $\tau \in \{1,2,\ldots,s\}$ 
\begin{align*}
\det \left(M_{\tau}(\encodSymbol,\edgeCodSymbol,\decodSymbol) \right) \neq 0.
\end{align*}
Thus, each $M_{\tau}(\encodSymbol,\edgeCodSymbol,\decodSymbol)$ is invertible.

\end{proof}
%
%

The following lemma improves upon the upper bound of Lemma~\ref{lem:1}
in the special case where the target function is linear over a
finite field.

\begin{lemma}\label{lem:2}
If $\Network$ is network 
with a linear target function $f$
over a finite field,
then
\begin{align*}
\codCap{\Network, f} & \le \underset{ C \in \cuts{\Network} } \min\ \card{C}.
\end{align*}
\end{lemma}

\begin{proof}
The same argument is used as in the proof of Lemma~\ref{lem:1},
except instead of using 
$R_{I_C,f} \ge 2$, we use the fact that $R_{I_C,f} = \card{\alphabet}$
for linear target functions.
\end{proof}

\begin{theorem}\label{Th:ModuloSumCodCap}
If $\Network$ is a network 
with a linear target function $f$ over finite field $\field{q}$,
then
$$
\linCodCap{\Network, f} = \underset{ C \in \cuts{\Network} }{\min} \card{C}.
$$
\end{theorem}%
\begin{proof}
%
%
We have
\begin{align*}
\linCodCap{\Network, f} 
&\le \codCap{\Network, f}\\
&\le \underset{ C \in \cuts{\Network} }{\min} \card{C}.
  & \Comment{Lemma~\ref{lem:2}}
\end{align*}
For a lower bound,
we will show that there exists an integer $n$ and an $(nc,n)$ linear code 
that computes $f$
with a computing rate
of $c= \displaystyle \underset{C \in \cuts{\Network} }{\min} \card{C}$.

From Lemma \ref{Lemma:inverse}, 
the matrix $I-F(y)$ in invertible over the ring $\polyRing{\field{q}}{x,y,w}$ 
and therefore also over $\polyRing{\field{q^n}}{x,y,w}$.
Since any minimum cut between the source $\source_{\tau}$
and the receiver $\receiver$ has at least $c$ edges,
it follows from~\cite[Theorem~2]{Koetter-Medard-IT-Oct03}%
\footnote{Using the implication $(1) \Longrightarrow (3)$ in~\cite[Theorem~2]{Koetter-Medard-IT-Oct03}.} 
that $\det\left( M_{\tau}(x,y,w) \right) \ne 0$
for every $\tau \in \left\{1,2,\ldots,s\right\}$.
From Lemma~\ref{Lemma:coeffChoice}, 
there exists an integer $n > 0$ and  vectors 
$\encodSymbol, \edgeCodSymbol, \decodSymbol$ over $\field{q^n}$ 
such that 
$M_{\tau}(\encodSymbol,\edgeCodSymbol,\decodSymbol)$ is invertible
for every $\tau \in \left\{1,2,\ldots,s\right\}$.
Since $f$ is linear, 
we can write $$f(u_1, \dots, u_s) = a_1 u_1 + \dots + a_s u_s.$$
For each $\tau \in \left\{1,2,\ldots,s\right\}$, let
\begin{align} \label{Eq:matrixChoise}
\hat{A}_{\tau}(\beta) =  a_{\tau} \left(M_{\tau}(\encodSymbol,\edgeCodSymbol,\decodSymbol)\right)^{-1} A_{\tau}(\beta).
\end{align}
%
If a linear code corresponding to the matrices $\hat{A}_{\tau}(\beta), B(\delta)$,
and $F(\gamma)$ is used in network $\Network$,
then  
the $c$-dimensional vector over $\field{q^n}$ computed by the receiver $\receiver$ is 
\begin{align*}
\psi & = \sum_{\tau=1}^{s} 
\sourceVec{\source_{\tau}} \hat{A}_{\tau}(\beta) 
(I - F(\gamma))^{-1} B(\delta)^{t} 
& \Comment{Lemma~\ref{Lemma:receivedSum} and linearity} \\
& = \sum_{\tau=1}^{s} \sourceVec{\source_{\tau}}  
a_{\tau} \left(M_{\tau}(\encodSymbol,\edgeCodSymbol,\decodSymbol)\right)^{-1} 
A_{\tau}(\beta) (I - F(\gamma))^{-1} B(\delta)^{t} & \Comment{\eqref{Eq:matrixChoise}} \\
%
& = \displaystyle \sum_{\tau=1}^{s} a_{\tau} \, \sourceVec{\source_{\tau}} &\Comment{\eqref{Eq:defnOfMtau}} \\
																											  & = \left( f\!\left(\sourceVec{\source_1}_1,\ldots,\sourceVec{\source_s}_1\right),\ldots,f\!\left(\sourceVec{\source_1}_c,\ldots,\sourceVec{\source_s}_c\right) \right)
\end{align*}
 which proves that the linear code achieves a computing rate of $c$.
\remove{
In other words, the receiver computes the linear sum 
(viewing $\sourceVec{\source_{\tau}}$ as a $cn$ length vector)
$$
\displaystyle \sum_{\tau=1}^{s} a_{\tau} \, \left(\sourceVec{\source_{\tau}}\right)_i
$$
which is equivalent to computing the target function 
$f$ for the $i$-th set of source messages for $i=1,\ldots,cn$.
}
\end{proof}
Theorem~\ref{Th:optimalityOfLinearCodes} 
below
proves the optimality of linear codes for computing linear target functions
in a single-receiver network. 
It also shows that the computing capacity of a network for a given target function 
cannot be larger than the number of network sources times the routing computing capacity for the same target function.
This bound tightens the general bound given in Theorem~\ref{Th:boundOnCodingGain} 
for the special case of linear target functions over finite fields.
Theorem~\ref{Th:linearTargetRoute} 
shows that this upper bound can be tight.

\begin{theorem}\label{Th:optimalityOfLinearCodes}
If $\Network$ is network 
with $s$ sources and
linear target function $f$ over finite field $\field{q}$,
then
$$
\linCodCap{\Network,f} = \codCap{\Network, f} \le s \ \routCap{\Network,f}.
$$
\end{theorem}

\begin{proof}
\begin{align*}
s\ \routCap{\Network,f} 
&\ge 
\underset{ C \in \cuts{\Network} }{\min} \card{C}     & \Comment{\eqref{Eq:routingCap} and Theorem~\ref{Th:routingCapacity}} \\
&\ge \codCap{\Network, f} & \Comment{Lemma~\ref{lem:2}} \\
 & \ge \linCodCap{\Network, f} \\
 & = \underset{ C \in \cuts{\Network} }{\min} \card{C}. \nonumber & \Comment{Theorem~\ref{Th:ModuloSumCodCap}}
\end{align*}
\end{proof}

We note that the inequality in Theorem~\ref{Th:optimalityOfLinearCodes}
can be shown to apply to certain target functions other than linear functions over finite fields,
such as the minimum, maximum, and arithmetic sum target functions.

\begin{theorem} \label{Th:linearTargetRoute}
For every $s$,
if a target function $f : \alphabet^s \longrightarrow \alphabet$ 
is linear over finite field $\field{q}$,
then there exists a network $\Network$ with $s$ sources,
such that
$$\linCodCap{\Network,f} = s \ \routCap{\Network,f}.$$
\end{theorem}
\begin{proof}
Let $\Network$ denote the network $\Network_{2,s}$ shown in
Figure~\ref{Fig:counterExample}.
Then
\begin{align*}
\linCodCap{\Network,f} & = 1 & \Comment{Theorem~\ref{Th:ModuloSumCodCap}} \\
\routCap{\Network,f}   &= \routCap{\Network} & \Comment{Theorem~\ref{Th:routingCapacity}}\\
 &= 1/s.  & \Comment{\eqref{Eq:routingCap}} 
\end{align*}
\end{proof}
\clearpage
\section{The reverse butterfly network} \label{Sec:butterfly}

In this section we study an example network which 
illustrates various concepts discussed previously in this paper and also
provides some interesting additional results for network computing.

\begin{figure}[ht]
\begin{center}
\psfrag{S}{$\hspace{-0.5cm}$ Source}
\psfrag{t1}{$\hspace{-0.8cm} \mbox{{\large Receiver $1$}}$}
\psfrag{t2}{$\hspace{-0.8cm} \mbox{{\large Receiver $2$}}$}
\psfrag{b}{{\large $(a)$ The multicast butterfly network}}
\psfrag{s1}{{\large $\source_1$}}
\psfrag{s2}{{\large $\source_2$}}
\psfrag{T}{{\large $\receiver$}}
\psfrag{a}{{\large $\hspace{1cm} (a)$ The butterfly network}}
\psfrag{b}{{\large $(b)$ The reverse-butterfly network}}
\scalebox{.6}{\includegraphics{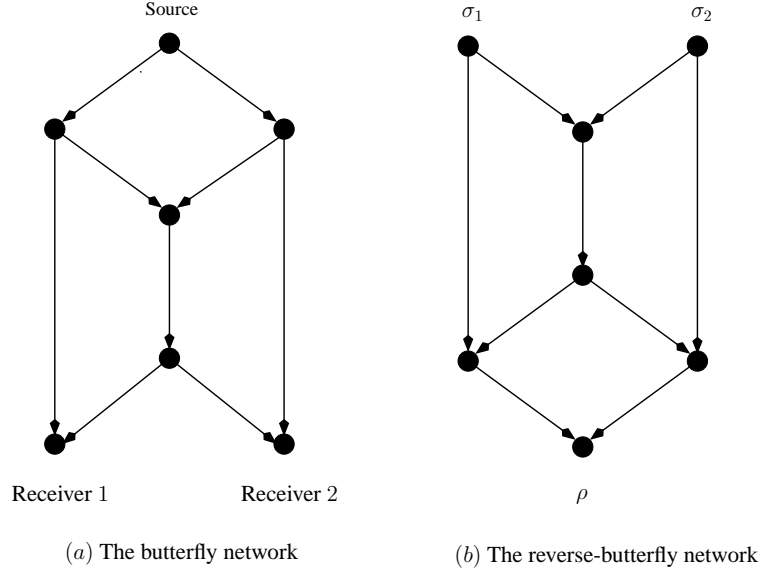}}
\end{center}
\caption{ The butterfly network and its reverse $\Network_3$.} \label{Fig:butterfly}
\end{figure}
The network $\Network_3$
shown in Figure~\ref{Fig:butterfly}(b)
is called the \textit{reverse butterfly network}.
It has $\sources = \{\source_1,\source_2\}$,
receiver node $\receiver$,
and is obtained by reversing the direction of all the edges of
the multicast butterfly network shown in Figure~\ref{Fig:butterfly}(a).


%
\begin{theorem} \label{Th:linearCoding}
The routing and linear computing capacities
of the reverse butterfly network $\Network_3$ 
with alphabet $\alphabet = \{0,1,\ldots,q-1\}$
and 
arithmetic sum target function
$f : \alphabet^{2} \longrightarrow \{0,1,\ldots,2(q-1)\}$ 
are
$$
\routCap{\Network_3,f} = \linCodCap{\Network_3,f} = 1.
$$
\end{theorem}
\begin{proof}
We have
\begin{align*}
\linCodCap{\Network_3,f} & = \routCap{\Network_3} & \Comment{Theorem~\ref{Th:linearCodingCapacity}} \\
		& = 1. & \Comment{\eqref{Eq:routingCap}}
\end{align*}
\end{proof}

\begin{remark}
The  arithmetic sum target function can be computed
in the reverse butterfly network at
a computing rate of $1$ using only routing
(by sending $\source_1$ down the left side and $\source_2$ down the right side of the graph).
Combined with Theorem~\ref{Th:linearCoding}, it follows that
the routing computing capacity is equal to $1$ for all $q \ge 2$.
\end{remark}

\begin{theorem} \label{Th:arithmeticSum}
The computing capacity 
of the reverse butterfly network $\Network_3$ 
with alphabet $\alphabet = \{0,1,\ldots,q-1\}$
and 
arithmetic sum target function
$f : \alphabet^{2} \longrightarrow \{0,1,\ldots,2(q-1)\}$ 
is
$$
\codCap{\Network_3,f} = \frac{2}{\log_q \left( 2q -1 \right)}.
$$
\end{theorem}

\begin{remark}
The computing capacity 
$\codCap{\Network_3,f}$ obtained in 
Theorem~\ref{Th:arithmeticSum} is a function of 
the coding alphabet $\alphabet$
(i.e. the domain of the target function $f$).
In contrast,
for ordinary network coding
(i.e. when the target function is the identity map),
the coding capacity and routing capacity are known to be independent of the coding alphabet used
~\cite{Cannons-Dougherty-Freiling-Zeger05}.
For the reverse butterfly network, 
if, for example, $q=2$,
then $\codCap{\Network_3,f}$ is approximately 
equal to $1.26$ and
increases asymptotically to $2$ as $q \rightarrow \infty$.
\end{remark}

\begin{remark}
The ratio of the coding capacity to the routing capacity for 
the multicast butterfly network with two messages was computed in
~\cite{Cannons-Dougherty-Freiling-Zeger05} to be $4/3$ 
(i.e. coding provides a gain of about $33\%$).
The corresponding ratio for the reverse butterfly network increases as a
function of $q$ from approximately $1.26$ (i.e. $26\%$) when $q=2$
to $2$ (i.e. $100\%$) when $q=\infty$.
Furthermore, in contrast to the multicast butterfly network, 
where the coding capacity is equal to the linear coding capacity,
in the reverse butterfly network
the computing capacity is strictly greater than the linear computing capacity. 
\end{remark}

\begin{remark}
Recall that capacity is defined as the supremum of a set of
rational numbers $k/n$ such that
a $(k,n)$ code that computes a target function exists.
It was pointed out in~\cite{Cannons-Dougherty-Freiling-Zeger05}
that it remains an open question whether the coding capacity of a network
can be irrational.
Our Theorem~\ref{Th:arithmeticSum}
demonstrates that the computing capacity
of a network 
(e.g. the reverse butterfly network) 
with unit capacity links can be 
irrational when the target function to be 
computed is the arithmetic  sum target function of the source messages. 
\end{remark}


%
\begin{figure}[ht]
\begin{center}
\psfrag{s1}{\mbox{\Large $\source_1$}}
\psfrag{s2}{\mbox{\Large $\source_2$}}
\psfrag{T}{\mbox{\Large $\receiver$}}
\psfrag{e1}{\mbox{\Large $x_1 \oplus x_2$}}
\psfrag{e2}{\mbox{\Large $y_2$}}
\psfrag{e4}{\mbox{\Large $x_1 \oplus x_2 \oplus y_2$}}
\psfrag{e3}{\mbox{\Large $x_1$}}
\psfrag{e5}{\mbox{\Large $y_1$}}
\psfrag{e7}{\mbox{\Large $x_1 \oplus x_2 \oplus y_1 \oplus y_2$}}
\psfrag{e6}{\mbox{\Large $x_2 \oplus y_2$}}
\scalebox{.8}{\includegraphics{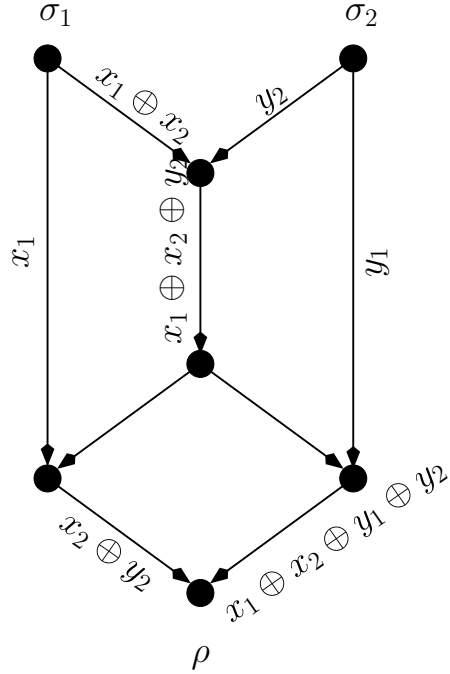}}
\end{center}
\caption{ The reverse butterfly network with a code that computes the mod $q$ sum target function.}
\label{Fig:reverse-butterfly-modulo-sum}
\end{figure}

The following lemma is used to prove Theorem~\ref{Th:arithmeticSum}.

\begin{lemma} \label{Lemma:modSum}
The computing capacity of the reverse butterfly network $\Network_3$
with $\alphabet = \{0,1,\ldots,q-1\}$
and the mod $q$ sum target function $f$ is 
$$
\codCap{\Network_3,f} = 2.
$$
\end{lemma}
\begin{proof}
The upper bound of $2$ on $\codCap{\Network_3,f}$ follows from 
~\cite[Theorem~II.1]{computing1}.
To establish the achievability part, 
let $k=2$ and $n=1$.
Consider the code shown in Figure~\ref{Fig:reverse-butterfly-modulo-sum}, 
where `$\oplus$' indicates the $\bmod$ $q$ sum.
The receiver node $\receiver$ gets $\sourceVec{\source_1}_1 \oplus \sourceVec{\source_2}_1$
and 
$\sourceVec{\source_1}_1 \oplus \sourceVec{\source_2}_1 \oplus \sourceVec{\source_1}_2 \oplus \sourceVec{\source_2}_2 $ 
on its in-edges, 
from which it can compute 
$\sourceVec{\source_1}_2 \oplus \sourceVec{\source_2}_2$.
This code achieves a rate of 2.
\end{proof}
\begin{proof}[Proof of Theorem~\ref{Th:arithmeticSum}:]
We have 
\begin{align*}
\codCap{\Network,f} \le 2/\log_q(2q-1). & \Comment{\cite[Theorem~II.1]{computing1}}
\end{align*} 
To establish the lower bound,
we use the fact the that arithmetic sum of two elements from 
$\alphabet = \{0,1,\ldots,q-1\}$ is equal to their $\bmod$ $2q-1$ sum.
Let the reverse butterfly network have alphabet 
$\hat{\alphabet} = \{0,1,\ldots,2(q-1)\}$.
From Lemma~\ref{Lemma:modSum} (with alphabet $\hat{\alphabet}$),
the $\bmod$ $2q-1$ sum target function can be computed in $\Network$
at rate $2$.
Indeed for every $n \ge 1$, 
there exists a $(2n,n)$ network code
that computes the mod $2q-1$ sum target function at rate $2$.
So for the remainder of this proof, let $k=2n$.
Furthermore, every such code using $\hat{\alphabet}$ can be ``simulated''
using $\alphabet$
by a corresponding $(2n,\ceil{n \log_q \left( 2q-1 \right)})$ code 
for computing the mod  $2q-1$ sum target function, as follows.
Let $n'$ be the smallest integer such that $q^{n'} \ge (2q-1)^{n}$, 
i.e., $n' = \ceil{n \log_q \left( 2q-1 \right)}$.
Let $g: \hat{\alphabet}^n \rightarrow \alphabet^{n'}$ be an injection
(which exists since $q^{n'} \ge (2q-1)^{n}$) and let the function 
$g^{-1}$ denote the inverse of $g$ on it's image $g(\hat{\alphabet})$.
%
%
%
%
Let $x^{(1)}, x^{(2)}$ denote the 
first and last, respectively, halves of the message vector $\sourceVec{\source_1}  \in \alphabet^{2n}$,
where we view $x^{(1)}$ and $x^{(2)}$ as lying in $\hat{\alphabet}^n$ 
(since $\alphabet \subseteq \hat{\alphabet}$).
The corresponding vectors $y^{(1)}, y^{(2)}$ for the source $\source_2$ are similarly defined.

Figure~\ref{Fig:arithmeticSum}  illustrates a $(2n,n')$ code for network $\Network$ using alphabet $\alphabet$
where `$\oplus$' denotes the $\bmod$ $2q-1$ sum.
Each of the nodes in $\Network$ 
converts each of the received vectors over $\alphabet$
into a vector over $\hat{\alphabet}$ using the function $g^{-1}$,
then performs coding in
Figure~\ref{Fig:reverse-butterfly-modulo-sum} 
over $\hat{\alphabet}$,
and finally converts the result back to $\alphabet$.
Similarly, the receiver node $T$ computes the component-wise arithmetic sum  of the source message vectors $\sourceVec{\source_1}$
and $\sourceVec{\source_2}$ using
\begin{align*}
&\sourceVec{\source_1} + \sourceVec{\source_2}\\
& = \big(g^{-1}(g(x^{(1)} \oplus x^{(2)} \oplus y^{(1)} \oplus y^{(2)})) \ominus g^{-1}(g(x^{(2)} \oplus y^{(2)})),\\
&\ \ \ \ \ \ \ \ g^{-1}(g(x^{(2)} \oplus y^{(2)}))\big)\\
& =(x^{(1)} \oplus y^{(1)}, x^{(2)} \oplus y^{(2)}).
\end{align*}
\begin{figure}[ht]
    \psfrag{s1}{\mbox{\Large $\source_1$}}
    \psfrag{s2}{\mbox{\Large$\source_2$}}
    \psfrag{T}{\mbox{\Large$\receiver$}}
    \psfrag{e1}{\large\mbox{$g(x^{(1)} \oplus x^{(2)})$}}
    \psfrag{e2}{\large\mbox{$g(y^{(2)})$}}
    \psfrag{e3}{\large\mbox{$g(x^{(1)})$}}
    \psfrag{e4}{\large\mbox{$g(x^{(1)} \oplus x^{(2)} \oplus y^{(2)})$}}
    \psfrag{e5}{\large\mbox{$g(y^{(1)})$}}
    \psfrag{e6}{\large\mbox{$g(x^{(2)} \oplus y^{(2)})$}}
    \psfrag{e7}{\large\mbox{$g(x^{(1)} \oplus x^{(2)} \oplus y^{(1)} \oplus y^{(2)})$}}
\begin{center}
\scalebox{.8}{\includegraphics{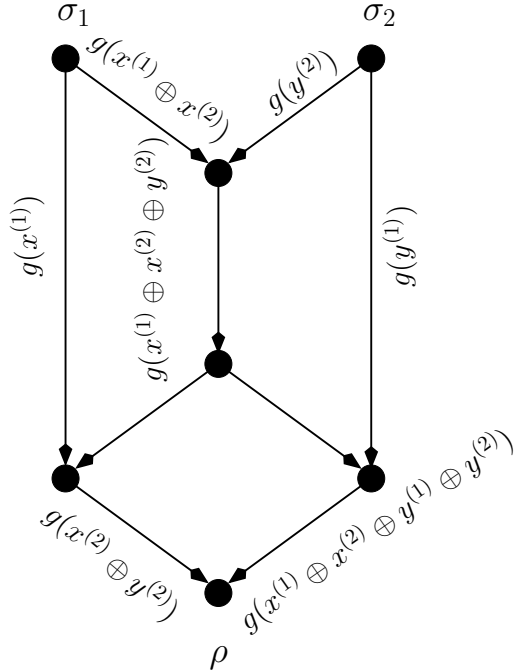}}
\end{center}
\caption{ The reverse butterfly network with a code that computes 
the arithmetic sum target function.
`$\oplus$' denotes $\bmod$ $2q-1$ addition.}
\label{Fig:arithmeticSum}
\end{figure}

For any $n \ge 1$, 
the above code computes the 
 arithmetic sum target function in $\Network$ at a rate of 
$$
\frac{k}{n'} = \frac{2n}{\ceil{n \log_q \left( 2q-1 \right)}}.
$$
Thus for any $\epsilon > 0$, 
by choosing $n$ large enough we obtain a code 
that computes the arithmetic sum target function, 
and which achieves a computing rate of at least 
$$
\frac{2}{\log_q \left( 2q-1 \right)} - \epsilon . 
$$
\end{proof}

\clearpage


%
\end{document}